\newcommand{\full}[1]{#1}
\newcommand{\itcs}[1]{}

\full{
\documentclass[11pt]{article}
\usepackage{geometry}
\geometry{verbose,letterpaper,tmargin=22mm,bmargin=22mm,lmargin=22mm,rmargin=22mm}
\usepackage{graphicx}
\usepackage{amsmath}
\usepackage{amssymb}
\usepackage[all]{xy}
\usepackage{bm}
\bibliographystyle{plain}
}

\itcs{
\documentclass[a4paper,UKenglish,cleveref, autoref, thm-restate]{lipics-v2021}
\bibliographystyle{plainurl}
}

\usepackage{tikz}

\full{

\newtheorem{theorem}{Theorem}[section]
\newtheorem{corollary}[theorem]{Corollary}
\newtheorem{lemma}[theorem]{Lemma}
\newtheorem{proposition}[theorem]{Proposition}

\newtheorem{definition}[theorem]{Definition}

\newtheorem{observation}[theorem]{Observation}

\newtheorem{claim}[theorem]{Claim}

\def\squarebox#1{\hbox to #1{\hfill\vbox to #1{\vfill}}}
\newcommand{\qed}{\hspace*{\fill}\vbox{\hrule\hbox{\vrule\squarebox{.667em}\vrule}\hrule}\smallskip}
\newenvironment{proof}{\noindent{\bf Proof:~~}}{\(\qed\)}
}

\newcommand{\eps}{\varepsilon} 
 
\newcommand{\hv}{\bar{v}}

\providecommand{\vv}{\varphi} 
\renewcommand{\vv}{\varphi} 

\title {Mechanism Design with Moral Bidders\full{\thanks{We thank the participants of the Israel Algorithmic Game Theory seminar for useful comments.}}} 
\full{
\author{Shahar Dobzinski\thanks{Weizmann Institute of Science. Work supported by BSF grant 2016192 and ISF grant 2185/19.}  \and Sigal Oren\thanks{Ben-Gurion University of the Negev. Work supported by BSF grant 2018206 and ISF grant 2167/19.}}
}
\itcs{
\author{Shahar Dobzinski}{Weizmann Institute of Science, Rehovot, Israel \and \url{https://sites.google.com/site/dobzin/} }{dobzin@gmail.com}{}{Work supported by BSF grant 2016192 and ISF grant 2185/19.}

\author{Sigal Oren}{Ben-Gurion University of the Negev, Beer-Sheva, Israel \and \url{https://sites.google.com/site/sigal3/}}{sigal3@gmail.com}{https://orcid.org/0000-0002-4271-7291}{Work supported by BSF grant 2018206 and ISF grant 2167/19.}

\authorrunning{S. Dobzinski and S. Oren} 

\Copyright{} 



\relatedversion{Full version is available on arXiv.} 

\acknowledgements{We thank the participants of the Israel Algorithmic Game Theory seminar for useful comments.}
}

\itcs{
	\EventEditors{Mark Braverman}
	\EventNoEds{1}
	\EventLongTitle{13th Innovations in Theoretical Computer Science Conference (ITCS 2022)}
	\EventShortTitle{ITCS 2022}
	\EventAcronym{ITCS}
	\EventYear{2022}
	\EventDate{January 31--February 3, 2022}
	\EventLocation{Berkeley, CA, USA}
	\EventLogo{}
	\SeriesVolume{215}
	\ArticleNo{100}
}

\begin{document}

		\maketitle
\begin{abstract}
	A rapidly growing literature on lying in behavioral economics and psychology shows that individuals often do not lie even when lying maximizes their utility. In this work, we attempt to incorporate these findings into the theory of mechanism design.
	
	We consider players that have a preference for truth-telling and will only lie if their benefit from lying is sufficiently larger than the loss of the others. To accommodate such players, we introduce $\alpha$-moral mechanisms, in which the gain of a player from misreporting his true value, comparing to truth-telling, is at most $\alpha$ times the loss that the others incur due to misreporting. Note that a $0$-moral mechanism is a truthful mechanism.
	
	We develop a theory of moral mechanisms in the canonical setting of single-item auctions within the ``reasonable'' range of $\alpha$, $0\leq \alpha \leq 1$. We identify similarities and disparities to the standard theory of truthful mechanisms. In particular, we show that the allocation function does not uniquely determine the payments and is unlikely to admit a simple characterization. In contrast, recall that monotonicity characterizes the allocation function of truthful mechanisms.
	
	Our main technical effort is invested in determining whether the auctioneer can exploit the preference for truth-telling of the players to extract more revenue comparing to truthful mechanisms. We show that the auctioneer can indeed extract more revenue when the values of the players are correlated, even when there are only two players. However, we show that truthful mechanisms are revenue-maximizing even among moral ones when the values of the players are independently drawn from certain identical distributions (e.g., the uniform and exponential distributions). 
	
	A by-product of our proof that optimal moral mechanisms are truthful is an alternative proof to Myerson's optimal truthful mechanism characterization in the settings that we consider. We flesh out this approach by providing an alternative proof that does not involve moral mechanisms to Myerson's characterization of optimal truthful mechanisms to all settings in which the values are independently drawn from regular distributions (not necessarily identical).
\end{abstract}

\thispagestyle{empty}
\clearpage
\setcounter{page}{1}

%

 \section{Introduction}

The backbone of mechanism design is the assumption that individuals will lie whenever doing so increases their profit. However, a rapidly growing literature in behavioral economics and psychology questions this assumption. Many experiments demonstrate that participants often do not lie or do not lie to the full extent in order to maximize their utility. One of the now classical experiments is described in \cite{fischbacher2013lies}: a participant is asked to roll a die in a cup and report the number to the experimenter (that cannot see which number the participant rolled). The payoff is determined according to the number reported, where higher numbers correspond to higher payoffs. Since the numbers are drawn from a known distribution, the aggregate deviation from truthful reports can be computed. A recent meta-analysis of such experiments by Abeler et al. \cite{abeler2019preferences} showed that participants only collect about 25\% of the potential gains from lying. Similar results were observed in other experiments (e.g., reporting the number of math problems solved \cite{mazar2006dishonesty}, sender-receiver games \cite{erat2012white,gneezy2005deception}).
 
Many models were suggested to explain this \emph{preference for truth-telling} of individuals. The meta-analysis of Abeler et al. \cite{abeler2019preferences} uses the data previously gathered together with new experiments to falsify many of these models. The models that survived this elimination combine a cost for lying and a desire to be perceived as honest by others.\footnote{These are also the main ingredients in two models that were published around the same time  \cite{khalmetski2019disguising,gneezy2018lying}.}  The cost of lying can be interpreted as a cognitive cost \cite{wang2010pinocchio,capraro2017does,christ2009contributions}. In contrast, the desire to be perceived as honest is guided by one's preference to keep his ``social identity'' aligned with his moral standards, as Gneezy et al. \cite{gneezy2018lying} suggest. Roughly speaking, according to social identity theory  (e.g., \cite{akerlof2000economics,benabou2011identity,tajfel2010social}), one's social identity is shaped by the way that others (even strangers) perceive him.

In a seminal paper, Gneezy \cite{gneezy2005deception} introduced the sender-receiver game to study deception. The sender-receiver game is a two-player game in which the sender recommends the receiver which alternative to choose: $A$ or $B$. Each alternative implies different rewards for the sender and the receiver. The sender is asked to tell the receiver which alternative is better for the receiver. While the sender is aware of the rewards, the receiver knows only the recommendation of the sender. The rewards that each alternative offered differ over treatments. However, alternative $A$  always maximized the receiver's reward, making $A$ the truthful recommendation. One remarkable finding is that participants lied less in treatments where their benefit from the lie was small relative to the loss of the other player.\footnote{Hurkens and Karnik \cite{hurkens2009would} provide an alternative explanation to the results of Gneezy \cite{gneezy2005deception}. They suggest that the results can be explained by assuming that some fraction of the population never lies and the rest lie whenever lying is beneficial. This criticism is refuted by Erat and Gneezy \cite{erat2012white} in a different experimental setup.} A possible explanation might be that the participants prefer alternative $A$ over alternative $B$ for other reasons.
For example, Fehr and Shmidt \cite{fehr1999theory} suggest that individuals have a preference for alternatives with a smaller difference in the payoff (i.e., inequality aversion). Charness and Rabin \cite{charness2002understanding} suggest that individuals prefer alternatives with higher social welfare (see \cite{fehr2006economics} for a comprehensive survey on such social preferences). 
Gneezy refutes those alternative explanations by showing that with the payoffs used in his experiments, the preference of $A$ over $B$ cannot be explained by the models of \cite{fehr1999theory} nor \cite{charness2002understanding}.



%

\subsection*{Moral Mechanisms}

So far, despite the large body of work demonstrating that individuals sometimes choose not to lie when it may be beneficial to them, relatively few theoretical works took this into consideration (e.g., \cite{matsushima2008role,kartik2014simple}). In this paper, we focus on players that rely on their moral standards when choosing their actions. In particular, such players have a preference for truth-telling (as demonstrated by \cite{abeler2019preferences}) and lie only if their benefit from lying is sufficiently greater than the loss of the other player(s) from lying  (as demonstrated by \cite{gneezy2005deception,erat2012white}).\footnote{The term moral is {somewhat loaded}, but we use it quite restrictively to refer to players that have these two preferences.} Due to the preference for truth-telling, both profit and loss are measured with respect to truth-telling. Broadly speaking, this paper joins the growing literature that is trying to connect algorithmic game theory and behavioral economics better (e.g., \cite{kleinbergTime,Kleinberg-mult-bias,kleinberg-soph,Kraft-time,Gravin:present-bias,EzraFF20,babaioff2018combinatorial,chen2021cursed}).



Concretely, in an $\alpha$-moral mechanism, the gain of a player from misreporting his true value, comparing to truth-telling, is at most $\alpha$ times the loss that the others incur due to misreporting (see Section \ref{sec-model} for precise definitions). Note that the case of $\alpha=0$ coincides with the standard definition of a truthful mechanism and that the larger the value of $\alpha$ the more the player ``cares'' about the consequences of lying for other players. In this paper we limit the discussion to $0 \leq \alpha \leq 1$. Other values of $\alpha$, very large and even negative, are unrealistic in many settings, but their analysis might be of mathematical interest. We leave studying these values to future work.


We stress that the behavior of moral players is very much different than the behavior of players that have externalities (such as altruistic or spiteful players, e.g., \cite{chen2008altruism,jehiel1996not}). The critical distinction here is that the preferences of players with externalities do not depend on their true value nor on the actions they may need to take to get the mechanism to output a preferred allocation. In contrast, moral players consider the utility of the other players only when they contemplate a lie that directly increases their own utility. 
As moral bidders try to follow their moral standards as much as possible, they only care about the utility of others when they are explicitly affecting it by taking a morally wrong action (e.g., lying). To give a concrete example, an altruistic player may lie to get the mechanism to output an allocation with higher social welfare even when his own utility remains the same or even decreases. In contrast, a moral player will not lie in this case.



Before diving into the technical details, it is worth discussing the applicability of moral mechanisms. Moral mechanisms are usually irrelevant in the case of competing firms. However, moral mechanisms are likely to be applicable when we have individuals that make decisions that directly influence their own payoff. In such cases, the social identity of the player is indeed shaped by his decisions. The player might therefore prefer to avoid lying, especially when the damage caused to others is large compared to the gains of the player.  Naturally, this tendency is usually stronger when the players know each other personally. Thus we believe that our model is relevant when auctioning resources within an organization, or auctioning among a group of individuals that share central traits such as religion or culture, and so on.
To illustrate this point, consider a hypothetical example of a department head who has to choose which PI will get some lab equipment, e.g., a fancy microscope. Each PI will naturally exaggerate the contribution of the new equipment to his lab, but we expect the exaggeration to be limited, in particular, if it is clear that another lab can make much better use of the equipment. 


\subsection*{A Theory of Moral Mechanisms}

We start with developing the foundations of moral mechanism design (Section \ref{sec-properties}). Our focus is the canonical setting of allocating a single item. We observe, as expected, that the set of allocation functions that can be implemented by moral mechanisms is a strict superset of the set of allocation functions that can be implemented by a truthful mechanism. In fact, for every $0\leq \alpha'<\alpha \leq 1$, the set of allocation functions of $\alpha$-moral mechanisms strictly contains the set of allocation functions of $\alpha'$-moral mechanisms. However, in sharp contrast to dominant-strategy mechanism design, where a dominant-strategy implementable allocation function uniquely determines the payment (up to a constant), multiple payment functions might be used to implement a given allocation function as a moral mechanism.



Next, we attempt to understand the structure of morally-implementable allocation functions. A key tool in the theory of mechanism design is the characterization of the set of dominant-strategy implementable allocation functions as the set of monotone functions. We provide some evidence that an analogous characterization does not exist for moral mechanisms in the sense that there are moral mechanisms whose allocation function can be arbitrary in many ``contiguous'' instances.


Nevertheless, we do obtain a simple condition that characterizes the set of moral mechanisms,  but this condition involves the payment function. It is well known that in a truthful mechanism, the price that the winner pays is only a function of the values of the others. That is, there are functions $p_{-1}, \ldots, p_{-n}$ and if player $i$ wins in the instance $(v_1,\ldots, v_n)$ then his payment is $p_{-i}(v_{-i})$. We refer to mechanisms that have this property as \emph{payment independent}. In a truthful mechanism, if player $i$'s profit is positive ($v_i-p_{-i}(v_{-i})>0$) then player $i$ will win the item and every player $j\neq i$ is not profitable ($v_j-p_{-j}(v_{-j})\leq 0$). In contrast, in a moral mechanism, several players might have a positive profit. This observation derives our characterization for every value of $\alpha$, stated here for the simple case of $\alpha=1$. A mechanism is a \emph{profit maximizer} if it is payment independent and the item is allocated to a player $i$ that maximizes $v_i-p_{-i}({v_{-i}})$, as long as this expression is not negative (importantly, more than one player might have a positive profit). In particular, for $\alpha=1$, we characterize the set of moral mechanisms as the set of profit maximizers.



\subsection*{The Main Result}

After understanding the basic principles of moral mechanism design, we go on to our main technical effort. Suppose that the players have positive values of $\alpha$. Is there a moral mechanism that generates more revenue than truthful mechanisms?\footnote{In our model, bidders care about the payoffs of the other bidders but not about the payoff of the auctioneer. This is quite common as often other bidders are treated as peers whereas the auctioneer is a part of a different social group. Consider, e.g., the example discussed above of PIs who are competing for a fancy microscope. The PIs probably care about the potential benefits of their colleagues from using the microscope but not so much about ``donating'' money to the institutional bureaucracy.}  In particular, can the auctioneer benefit from designing the auction in a way that will effectively ``amplify'' the players' values of $\alpha$?\footnote{\label{footnote-truthful}This might be done by making the consequences of the actions of the players more explicit. For example, in a different context, Read et al. \cite{read2017value} asked participants to choose between taking \$5 to their pocket and donating \$5 to the red cross. When the choices were made more explicitly -- ``take the money and the red cross does not get the money`` and ``donate the money to the red cross'' -- the percentage of donations to the red cross increased.}


We consider the case where the values $(v_1,\ldots, v_n)$ are drawn from a distribution $\mathcal F$ in a correlated way (Section \ref{sec-correlated}). We show that for \emph{any} $0< \alpha \leq 1$ there is a distribution $\mathcal F$ such that an $\alpha$-moral mechanism extracts more revenue than its optimal truthful counterpart by a multiplicative constant factor. In the other direction, we show that truthful mechanisms extract at least half the revenue of the revenue-maximizing $\alpha$-moral mechanism.

In Section \ref{sec-independent} we study the classic setting of revenue maximization with {$n$} players where the value of each player $i$ is drawn independently from a known and identical distribution $\mathcal F$. For convenience, we assume that the valuations are discrete, i.e., there is some $\eps>0$ such that the support of $\mathcal F$ is $\{0,\eps, 2\eps, \ldots, 1-\eps, 1\}$. Our main result shows that for certain distributions, including the uniform and exponential distributions, the optimal deterministic moral mechanism is in fact truthful even when $\alpha$ is as large as $1$. In other words, the auctioneer cannot exploit the ``moral standards'' of the players to increase his revenue.

Ideally{, to prove our main result}, we would like to extend Myerson's celebrated characterization of optimal mechanisms \cite{myerson1981optimal} to hold for profit maximizers (moral mechanisms with $\alpha=1$) and not just for truthful mechanisms.\footnote{Note that proving that optimal $1$-moral mechanisms are truthful immediately implies that optimal $\alpha$-moral mechanisms for any $0\leq \alpha <1$ are truthful as well.} However, this approach fails almost immediately: Myerson's first step is to characterize the payment function in terms of the allocation function. Given the characterization, Myerson's proof completely ignores the payment functions and continues to find the allocation function that maximizes revenue. This approach inherently fails for moral mechanisms since, as we show, the allocation function of $1$-moral mechanisms is ``almost'' unconstrained. Thus a different approach to characterizing optimal moral mechanisms is needed. Therefore, instead of ``ignoring'' the payment functions almost immediately and focusing on the allocation function, we develop a very different proof technique that puts the payment functions in the center. By subtle manipulations, we can characterize the payment functions of an optimal moral mechanism. We discover that the optimal moral mechanism is truthful, thus proving the optimality of Myerson's mechanisms also among moral mechanisms. 

As a by-product, we obtain an alternative proof of Myerson's theorem for the settings that we consider. \itcs{In the full version we}\full{We} flesh out our approach by providing\full{(Section \ref{sec-truthful})} an alternative proof to Myerson characterization of optimal \emph{truthful} mechanisms when the values of the players are independently drawn from (not necessarily identical) distributions that satisfy Chung and Ely's discrete analogue of the regularity condition \cite{chung2007foundations}. Our alternative proof is not shorter than Myerson's but is more ``direct''. We -- and this is obviously a subjective statement -- find it easier to digest since our proof gradually ``exposes'' the mechanism rather than obtaining it at once via Myerson's spectacular but somewhat opaque proof. We hope that this proof will find more applications beyond characterizing optimal moral mechanisms.

\subsection*{Open Questions and Future Directions}

It is well documented that in practice bidders sometimes do not bid truthfully, even when truth telling is a dominant strategy. The literature underlines the importance of the specific implementation in motivating truthful bidding (e.g., ascending vs. sealed bid auctions). Our work raises the possibility of using implementations that exploit behavioral biases to encourage truthful bidding in auction settings (see footnote \ref{footnote-truthful} and the discussion that leads to it). In that light, one implication of our paper is that in some settings the revenue that the auctioneer can extract does not increase with implementations that fully exploits the tendency of bidders to truth telling. It is a fascinating future direction to further explore this issue which resides on the border of behavioral economics and auction design.

We also leave several questions open on the immediate level: can moral mechanisms raise more revenue than their truthful counterparts when the distributions are not identical or not even regular? The proofs that moral mechanisms sometimes obtain more revenue and that truthful mechanisms are optimal for certain independent distributions are both quite subtle and make use of several novel technical constructs. We therefore believe that progress on this front will reveal more exciting insights. Furthermore, even considering truthful mechanisms, can the new proof technique for characterizing optimal mechanisms be used in other settings to obtain new results about truthful mechanisms, or maybe simplify existing ones?

Another promising avenue for future research is exploring moral mechanisms in other settings: can we characterize them? Are they more powerful? Even if they do not have more power, perhaps there are computationally efficient moral mechanisms for problems in which no computationally efficient truthful mechanisms exist.

\section{Model}\label{sec-model}


This paper studies single item auctions. Let $N=\{1,\ldots,n\}$ be a set of players. The private information of player $i$ is his value $v_i$ for the item. The set of possible values of player $i$ is $\mathbb R^+$. An allocation function $f: (\mathbb R^+)^n \rightarrow N \cup \{0\}$ receives the bids of all players and decides whether to allocate the item and to which player. A direct mechanism is composed of an allocation function $f$ and a payment function $p: (\mathbb R^+)^n \rightarrow \mathbb R^n$ that determines the payment of each player. All mechanisms in this paper are deterministic.

We study moral mechanisms for single item auctions that obey the following standard properties:
\begin{enumerate}
	\item \textbf{Individual Rationality:} for every player $i$ and $(v_1,\ldots, v_n)$, $v_i \cdot \delta(i=f(v_1,\ldots, v_n))\geq p(v_1,\ldots, v_n)_i$ (i.e., a player does not pay more than his value). Here, $\delta(\cdot)$ is a function that receives a boolean condition and returns $1$ if the condition holds and $0$ otherwise.
	\item \textbf{No Positive Transfers:} For every player $i$ and $(v_1,\ldots, v_n)$, $p(v_1,\ldots, v_n)_i \geq 0$.
\end{enumerate}

Note that the two properties imply that the payment of a player that does not get the item is $0$.
Define the profit of player $i$ with true value $v_i$ for reporting value $v'_i$: 
\begin{align*}
\pi_i(v'_i,v_{-i}|v_i) = v_i \cdot \delta(i=f(v'_i,v_{-i}))-p(v'_i,v_{-i})_i
\end{align*}
%
%
\begin{definition}[Moral Mechanisms]
Consider an individually rational mechanism with no positive transfers for single item auction $M=(f,p)$. 
$M$ is \emph{$\alpha$-moral} if in every instance $(v_1,\ldots, v_n)$, for every player $i$ and $v'_i$ such that $\pi_i(v'_i,v_{-i}|v_i)  - \pi_i(v_i,v_{-i}|v_i)>0$:
  \begin{align*}
  \pi_i(v'_i,v_{-i}|v_i)  - \pi_i(v_i,v_{-i}|v_i)  \leq  \alpha \cdot \sum_{j \neq i} \big(   \pi_j(v_i,v_{-i}|v_j)  - \pi_j(v'_i,v_{-i}|v_j) \big) 
 \end{align*}
when $\alpha=1$ we simply say that $M$ is \emph{moral}.
\end{definition}
For example, when $\alpha=1$ and all players report truthfully, if player $i$ misreports its value then its profit might indeed increase (comparing to truth-telling), but the sum of profits of the other players decreases by at least the amount that the profit of player $i$ increased.

Observe that any mechanism that is $\alpha$-moral is also $\alpha'$-moral for $\alpha'>\alpha$. In particular, note that a $0$-moral mechanism is simply a truthful mechanism. To simplify the presentation, we will sometimes limit ourselves to the somewhat ``cleaner'' case of $\alpha=1$.  We note that our main result that shows that among all $1$-moral mechanisms, the revenue-maximizing mechanism is truthful ($0$-moral) directly implies that for every $0\leq \alpha\leq 1$ the revenue-maximizing mechanism among all $\alpha$-moral mechanisms is truthful.

Note that truth-telling is an equilibrium of $\alpha$-moral mechanisms if every player $i$  believes that the other players report truthfully and player $i$ will not misreport its value if its gain is less than $\alpha$ times the loss of the other players from this misreport\footnote{This is similar to, e.g., direct revelation Bayes-Nash incentive compatible mechanisms: truth-telling is an equilibrium if all players agree on the underlying distributions from which the values are drawn and assume that the other players bid their true types.}. In this
paper we analyze the truth-telling equilibrium of $\alpha$-moral mechanisms. 

Our definition of moral mechanisms is inspired by \cite{gneezy2005deception,gneezy2018lying} who only consider two-player games. Two natural extensions come to mind when trying to extrapolate to $n$-player games: a player weighs his lie against the sum of losses of the other players, or against the maximal loss of a single player. In our single item auction setting, since the payoff of each player that does not receive the item is $0$, these two extensions coincide.


\section{Properties of $\alpha$-Moral Mechanisms}\label{sec-properties}
We start with  showing that the relaxation to $\alpha$-moral mechanisms strictly extends the set of allocation functions that can be implemented: 
\begin{claim} \label{claim-allocation-not-monotone}
For any $\alpha >0$ there exists an allocation function that can be implemented with an $\alpha$-moral mechanism but not in dominant strategies.
\end{claim}
\begin{proof}
Consider the following mechanism:
		\begin{itemize}
		\item If $v_1 \geq 1$, $v_1\neq 1 + \frac {\alpha}{10}$ and $v_2\geq \frac{1}{10}$ -- allocate to player $1$ at price $1$.
		\item If $v_1=1 + \frac {\alpha}{10}$ -- allocate to player $2$ at price $0$.
		\item Otherwise, do not allocate the item.
	\end{itemize}
Note that the allocation function is not monotone, hence the mechanism cannot be implemented in dominant strategies, even with different payments. To see that the allocation function is not monotone, consider the case where $v_1=1$ and $v_2=\frac {2}{10}$. The item is allocated to player $1$. However, when $v_1=1+ \frac {\alpha}{10}$ and $v_2$ remains the same, the item is allocated to player $2$. 

We now show that the mechanism is $\alpha$-moral. First, observe that truth telling is a dominant strategy for player $2$: when he gets the item his payment is always zero and whether he gets the item or not is not a function of his own value. As for player $1$, we first observe that if $v_2<\frac {1}{10}$, player $1$ does not get the item regardless of his report. Now, assuming that $v_2\geq \frac {1}{10}$ we distinguish between two cases:

\begin{itemize}
	\item When player $1$ gets the item he pays $1$. Thus, when $v_1<1$ player $1$ should bid his true value. Similarly, when $1+ \frac {\alpha}{10} \neq v_1>1$ his profit for reporting his true value is $v_1-1\geq 0$.
	\item If $v_1=1 + \frac {\alpha}{10}$, player $2$ gets the item. If $v_2\geq \frac {1}{10}$, the only way for player $1$ to have a positive profit is by reporting a value that is at least $1$ but not $1 + \frac {\alpha}{10}$. In this case player $1$ gets the item and pays $1$, so the profit of player $1$ increases to $\frac {\alpha}{10}$ from $0$. However, notice that player $2$ gets the item and pays nothing, so his profit is at least $v_2-0\geq \frac {1}{10}$. Thus the profit of player $1$ indeed increases by reporting a different value, but not by more than $\alpha$ times the decrease of the profit of player $2$. Hence the mechanism is $\alpha$-moral.
\end{itemize}
\end{proof}

\subsection{The Payment Functions of $\alpha$-Moral Mechanisms}
Truthful mechanisms have the property that the price that a player pays when it wins the item does not depend on the value that it reports. We say that a (not necessarily truthful) mechanism with this property is \emph{payment independent}. We observe that $\alpha$-moral mechanisms are payment independent.
\begin{definition}[payment independence]
A mechanism $M=(f,p)$ is \emph{payment independent} if the price of every winning player $i$ does not depend on his own value. I.e., for every $v_i, v'_i, v_{-i}$, if player $i$ wins in both instances $(v_i, v_{-i})$ and $(v'_i, v_{-i})$ then $p(v_i,v_{-i})_i=p(v'_i,v_{-i})_i$.
\end{definition}

\begin{observation}\label{obs-payment-independent}
Every $\alpha$-moral mechanism is payment independent.
\end{observation}
\begin{proof}
Consider an $\alpha$-moral mechanism $M=(f,p)$. Fix some $v_{-i}$ and assume that there are $v_i,v'_i$ such that the item is allocated to player $i$ in both instances $(v_i,v_{-i})$ and $(v'_i,v_{-i})$. Toward contradiction and without loss of generality, suppose that $p(v_i,v_{-i})_i> p(v'_i,v_{-i})_i$.
Note that if player $i$'s value is $v_i$, reporting $v'_i$ increases his profit by reducing his payment. Observe that in both instances all the other players have a zero profit  since they do not win anything and hence pay nothing. Thus, the mechanism is not moral since player $i$ can increase his profit by reporting $v'_i$ instead of $v_i$ without reducing the profit of any of the other players.
\end{proof}

The theory of dominant-strategy mechanisms states that in single item auctions an implementable allocation function uniquely defines the prices. In contrast, when the mechanism is $1$-moral, the same allocation function may be implemented by different payment functions. For example, we show that:

\begin{claim} \label{clm-arb-prices}
Consider the allocation function that allocates the item to the player with the highest value. For any $0\leq c\leq 1$, the mechanism that allocates the item to the player with the highest value and charges him $c\cdot v_j$, where $v_j$ is the second-highest value, is $1$-moral.
\end{claim}
\begin{proof}
	Let $i$ be the player with the highest value. This player is allocated the item, and his profit is non-negative $v_i-c\cdot v_j$ ($v_j$ is the second highest value). Any other report will either make his profit $0$ or will not affect his profit, so player $i$'s dominant strategy is to report his true value.
	
	Consider some player $k\neq i$. In truth-telling, his profit is $0$. Any report below $v_i$ will not make him win the item, but player $k$ can win the item by reporting a value that is bigger than $v_i$. In this case his profit changes to $v_{k}-c\cdot v_i$. This profit might be negative, but in any case it cannot be that the profit of player $k$ increased more than the decrease in profit of player $i$. To see why, observe that $(1+c) \cdot v_i \geq (1+c)\cdot v_j$. This implies that $v_i - c\cdot v_j \geq v_{j}-c\cdot v_i \geq v_{k}-c\cdot v_i$ as required.
\end{proof}

Note that in the above example the higher the value of $c$ the more revenue that the mechanism generates in every instance. In particular, among all possible choices of $c$, $c=1$ gives the mechanism with the highest revenue in that family. Interestingly, when $c=1$ the mechanism is just a second-price auction, which is obviously a dominant strategy mechanism. We show that this is not a coincidence: if $f$ is an allocation function that can be implemented as an $\alpha$-moral mechanism by several payment functions and one of the implementations is a dominant strategy implementation, then in every instance the revenue that the dominant strategy implementation generates is at least as high as the other implementations.
\begin{proposition}\label{prop-truthful-implementation}
Consider a dominant-strategy mechanism $M=(f,p)$ and an $\alpha$-moral mechanism $M'=(f,p')$ (both mechanisms implement the same social choice function), both are individually rational and have no positive transfers. Then, $p' \leq p$. 
\end{proposition}
\begin{proof}
	Recall that for any instance $v=(v_1,\ldots,v_n)$ and player $i\neq f(v)$ we have that $p(v)_i=p'(v)_i=0$. Assume towards contradiction that there exists $v_i,v_{-i}$ such that $f(v_i,v_{-i})=i$ but $p'(v_i,v_{-i})_i>p(v_i,v_{-i})_i$. Let $v'_i$ be such that $p'(v_i,v_{-i})_i>v'_{i}>p(v_i,v_{-i})_i$. Since $M$ is a dominant-strategy mechanism we have that $f$ is monotone and $p(v_i,v_{-i})_i$ is the critical value of player $i$. This implies that player $i$ wins the item in the instance $(v'_i,v_{-i})$ in both mechanisms. 	
Since every moral mechanism is payment independent (Observation \ref{obs-payment-independent}), in $M'$ player $i$ pays $p'(v'_i,v_{-i})_i = p'(v_i,v_{-i})_i$. Hence, we get that in this case $v'_i < p'(v'_i,v_{-i})_i$ in contradiction to individual rationality.
\end{proof}

\subsection{Characterizations and non-Characterizations of Moral Mechanisms}
Working directly with the definition of moral mechanisms in order to construct and analyze moral mechanisms is not always easy. Thus, we identify a certain natural class of moral mechanisms, called \emph{profit maximizers}, that is easier to work with. Fortunately, we will show that $\alpha$-profit maximizers are equivalent to $\alpha$-moral mechanisms.

\begin{definition}[Potential Profit]
Let $M=(f,p)$ be a payment independent mechanism. Fix $v_{-i}$. Let $p_{-i}(v_{-i})$ be the payment of player $i$ if there exists some value $v'_i$ such that $f(v'_i,v_{-i})=i$. If such $v'_i$ does not exist then $p_{-i}(v_{-i})$ is defined to be $\infty$.
$p_{-i}(v_{-i})$ is the \emph{potential payment} of player $i$ and the \emph{potential profit} of player $i$ is $v_i-p_{-i}(v_{-i})$.
\end{definition}

	\begin{definition}[$\alpha$-Profit Maximizers]
		Let $M=(f,p)$ be an individually rational payment independent mechanism with no positive transfers. $M$ is a \emph{profit maximizer} if when the item is allocated, it is allocated to a player with a maximal non-negative potential profit. Furthermore, the item is allocated whenever the potential profit of at least one of the players is positive.

		A mechanism is \emph{$\alpha$-profit maximizer} for $0\leq \alpha \leq 1$ if in addition the following holds: for every $(v_1,\ldots, v_n)$ such that the item is allocated to player $i$,		
		for every player $j\neq i$:
		\begin{align*}
		v_j - p_{-j}(v_{-j}) \leq \alpha \cdot (v_i - p_{-i}(v_{-i}))
		\end{align*}
		
	\end{definition}

We refer to $1$-profit maximizers as profit maximizers. Note that, every dominant-strategy mechanism is a profit maximizer in the sense that if the potential profit of player $i$ is positive, then his value is above his critical value (which is his potential payment). However, $\alpha$-profit maximizer (for $\alpha>0$) are more general in the sense that there might be multiple players with positive potential profits.

\begin{claim}
An individually rational mechanism with no positive transfers is $\alpha$-moral if and only if it is $\alpha$-profit maximizing. 
\end{claim}
\begin{proof}
Consider an $\alpha$-profit maximizing mechanism $M$. The profit of the winner is always non-negative. By no positive transfers, the profit of all other players is $0$. Observe that the winner will not misreport his value since
$M$ is payment independent and hence the winner's profit can only be reduced to $0$ (when he reports a value that makes him lose the item). As for the losers, let $i=f(v_1,\ldots,v_n)$ and consider a loser $j$. Player $j$ can misreport and change his profit from $0$ to his potential profit $v_j - p_{-j}(v_{-j})$. However, in this case the winner's profit decreases by his potential profit $(v_i - p_{-i}(v_{-i}))$. Since this is an $\alpha$-profit maximizer, we have that $v_j - p_{-j}(v_{-j}) \leq \alpha (v_i - p_{-i}(v_{-i}))$, as required from an $\alpha$-moral mechanism. 
Finally, in the case that the item is not allocated (i.e., 
$0=f(v_1,\ldots,v_n)$) we have that the potential profit of any loser $j$ is non-positive (i.e., $v_j - p_{-j}(v_{-j})\leq 0$), thus player $j$ cannot increase his profit by misreporting. Hence the mechanism is $\alpha$-moral. 

In the other direction, consider an individually rational $\alpha$-moral mechanism with no positive transfers. Here, the only possible effect of misreporting is that the profit of the winner decreases by its potential profit and the profit of a loser increases by its potential profit. However, in an $\alpha$-moral mechanism, $\alpha$ times the potential profit of the winner must be at least as large as the potential profit of any loser, otherwise misreporting is beneficial. This implies that the mechanism is $\alpha$-profit maximizer.
\end{proof}

\full{In Proposition \ref{prop-lattice} in the appendix}\itcs{In the full version} we rely on the
equivalence of $\alpha$-profit maximizers and $\alpha$-moral mechanisms 
%
to prove that the payment functions that implement a given allocation function form a lattice-like structure.

A basic building block in mechanism design is the characterization of monotone allocation functions as the set of implementable functions in single parameter domains. Can we hope to obtain a similar characterization for moral mechanisms? We have already seen that when the mechanism is moral the allocation might not be monotone (Claim \ref{claim-allocation-not-monotone}). One might hope that a different characterization of the morally implementable allocation functions exists. \full{In Claim \ref{clm-no-char}}\itcs{In the full version} we show that essentially there is no simple characterization for $1$-moral mechanisms. This is done by providing a family of allocation functions that are implementable using $1$-moral mechanisms and their allocation can be arbitrary in a wide range of values, when $v_1+ v_2\geq 1$ and $v_1,v_2\leq 1$. In general, we do not know whether a characterization of the allocation functions of $\alpha$-moral mechanisms exists. \full{However, Claim \ref{clm-rule} in the appendix provides}\itcs{However, in the full version we provide} a simple tool to rule out the implementability of some functions.

\section{Revenue Maximization with Independent Distributions}\label{sec-independent} 

Our main technical effort is invested in understanding whether the auctioneer can take advantage of the moral standards of the players to generate more revenue than the optimal truthful mechanism. Section \ref{sec-correlated} shows that the answer is yes when the values are drawn in a correlated manner. In contrast, we now show that moral mechanisms do not generate more revenue than truthful ones when the values are drawn i.i.d. from some set of distributions (which includes the uniform and exponential distributions).

Recall that the allocation function of moral mechanisms does not uniquely determine the payments. However, by Proposition \ref{prop-truthful-implementation}, if one of the implementations of a moral mechanism is a dominant strategy one then this implementation gives the highest revenue in every instance. By itself, of course, this proposition does not guarantee that the optimal moral mechanism is truthful. For example, the mechanism given in Section \ref{sec-correlated} is moral and is not dominant strategy (and generates more revenue than any dominant strategy mechanism). 

Alternatively, one could hope to extend Myerson's proof to hold for moral mechanisms as well. However, the apparent lack of structure of moral mechanisms (e.g., no condition analogous to monotonicity) is a major obstacle to implementing this approach. Instead, we develop an alternative technique that characterizes optimal moral mechanisms by directly characterizing the payment functions. Since moral mechanisms are profit maximizers, the payment functions determine the allocation function up to tie-breaking. Our proof holds for truthful mechanisms, which are $0$ moral, as well. In fact, whereas for moral mechanisms we need some additional constraints on the distributions, for the simpler case of truthful mechanisms we provide \itcs{in the full version} an alternative proof to Myerson's characterization when values are drawn from independent (not necessarily identical) regular distributions \full{(Section \ref{sec-truthful})}.


For convenience, we assume that the support of the distributions is discrete. We use the following extension of virtual valuations and regularity for discrete distributions: 

\begin{definition}[Chung and Ely \cite{chung2007foundations}]\label{def-virtual-value}
Fix some $\eps>0$. Let $\mathcal F$ be a distribution with support $\{0,\eps, 2\eps, 3\eps,..., 1-\eps, 1\}$. Let $F,f$ be its CDF and PDF, respectively. Let the \emph{virtual valuation} of $\mathcal F$ be ${\varphi(v)= v-\eps\cdot \dfrac{1-F(v)}{f(v)}}$. The distribution $\mathcal F$ is \emph{regular} if $\varphi(\cdot)$ is non-decreasing.\footnote{In fact, Chung and Ely assume that the support is $\{a,a+\eps, a+2\eps, 3\eps,..., b-\eps, b\}$ for some $b>a$. We use $a=0$ and $b=1$ to slightly simplify the proof, but our proof can easily be extended to support all other values of $a$ and $b$.}
\end{definition}

We identify a class of regular distributions for which the revenue-maximizing moral mechanism is truthful:

\begin{definition}
Fix some $\eps>0$. Let $\mathcal F$ be a distribution with support $\{0,\eps, 2\eps, 3\eps,..., 1-\eps, 1\}$. Let $F,f$ be its CDF and PDF, respectively. We say that $\mathcal F$ is \emph{standard} if for every $v>v'$: 
	\begin{align*}
	v -\eps \cdot \frac {(1-F(v))} {f(v)}\geq  v'-\eps  \cdot \frac {(1-F(v))} {f(v')} 
	\end{align*}
\end{definition}

Note the syntactic similarity of the definitions: the definition of a regular distribution is obtained by replacing the term $F(v)$ in the RHS of the definition of a standard distribution with $F(v')$. In particular, any standard distribution is also regular since for $v\geq v'$, $(1-F(v))\leq (1-F(v'))$. We show that:
\begin{proposition} \label{prop-standad}
The following distributions are standard:
\begin{enumerate}
\item The uniform distribution on $[a,b]$, for some $0\leq a <b$.
\item The exponential distribution.
\item\label{cond-increasing-pdf} All distributions such that for every $v \geq v'$ in the support, $f(v)\geq \frac{f(v')}{1+f(v')}$.
\end{enumerate}
\end{proposition}
\begin{proof}
		We first prove that distributions that satisfy condition \ref{cond-increasing-pdf} are standard. To this end it suffices to show that if condition \ref{cond-increasing-pdf} holds then:
		\begin{align*}
		f(v') \cdot \eps \cdot (1-F(v)) \leq f(v) \cdot \eps  \cdot(1-F(v)) + f(v)\cdot f(v') \cdot (v-v')
		\end{align*}
		First note that $v\geq v'+\eps$, thus $v-v'\geq \eps$. Hence, a stronger condition is:
		\begin{align*}
		f(v') \cdot  (1-F(v)) \leq f(v)  \cdot(1-F(v)) + f(v)\cdot f(v')
		\end{align*}
		Since $(1-F(v)) >0$, we get that a stronger condition is:
		\begin{align*}
		f(v') \leq f(v)  + f(v)\cdot f(v')
		\end{align*}
		It is easy to see that the last inequality holds for $f(v)\geq \frac{f(v')}{1+f(v')}$.
		
		It is obvious that the uniform distribution satisfies condition \ref{cond-increasing-pdf}. We now show that the exponential distribution is standard. Consider the exponential distribution with support in $[0,\eps, 2\eps, \ldots, 1]$: $F(x)=1-e^{-\frac x \eps}$ for $x<1$ in the support and $F(1)=1$. Note that $f(x)=F(x)-F(x-\eps) = e^{\frac {-x+\eps} \eps} - e^{-\frac x \eps }=(e -1)\cdot e^{-\frac x \eps}$. We need to show that:
		\begin{align*}
		(e-1)\cdot e^{-\frac {v'} \eps} \cdot \eps \cdot (e^{-\frac {v} \eps}) &\leq (e-1)\cdot e^{-\frac {v} \eps} \cdot \eps  \cdot (e^{-\frac {v} \eps}) + (e-1)\cdot e^{-\frac {v} \eps}\cdot (e-1)\cdot e^{-\frac {v'} \eps} \cdot (v-v')\\
		e^{-\frac {v'} \eps} \cdot \eps  &\leq \eps  \cdot (e^{-\frac {v} \eps}) +  (e-1)\cdot e^{-\frac {v'} \eps} \cdot (v-v')
		\end{align*}
		which holds for any $v\geq v'$. 
\end{proof}

\begin{theorem}\label{thm-standard}
Fix $0 \leq \alpha \leq 1$. Let $M$ be an $\alpha$-moral mechanism for selling one item to $n$ players. Suppose that the values of the players are drawn i.i.d. from a distribution $\mathcal F$ and that $\mathcal F$ is standard. Then, there is a \emph{dominant-strategy} mechanism $M'$ with revenue that is at least as large as the revenue of $M$.
\end{theorem}

\full{The full proof of the theorem is given in Section \ref{sec-thm-standard}.}
Since the proof is quite involved, we now provide a sketch for the $2$-player case to gently introduce the reader to the main building blocks of the proof. 
 The proof for $n$ players is similar except for the induction hypothesis, which is more subtle in the $n$ player case. \itcs{This proof can be found in the full version.}

\subsection{An Outline of the Proof for 2 Players}

Our goal is to prove that the optimal moral auction for two players is in fact truthful. Specifically, we will prove that this is a second-price auction with a reserve price. The reserve price equals the monopolist price of $\mathcal F$, denoted $x=\arg\max_pp\cdot (1-F(p-\eps))$ (recall that the monopolist price of $\mathcal F$ is the take-it-or-leave-it price that maximizes the revenue in an auction with only one player whose value is drawn from $\mathcal F$). The main challenge is to prove the following lemma:

\begin{lemma}\label{lemma-sketch-main}
For every player $i$ and value $v\geq x$ of the other player we have that $p_{-i}(v)\geq v$.
\end{lemma}
That is, if the other player's value is $v$ for some $v\geq x$, then the potential payment of player $i$ (i.e., the payment of player $i$ if he wins the item) is at least $v$.

Before sketching the proof of Lemma \ref{lemma-sketch-main}, we apply the lemma and show that the mechanism is truthful.
%
 For the mechanism to be truthful, it is sufficient to show that for every player $i$ and value $v$ we have that $p_{-i}(v)\geq v$. The lemma 
 guarantees this for $v\geq x$ and in the next paragraph we prove this for $v<x$.
  This is sufficient since it implies that in every instance the potential payment of the player with the lower value is at least the value of the higher value player and hence there could be at most one player with positive potential profit.
  

Thus, assuming Lemma \ref{lemma-sketch-main}, it remains to show that for $v<x$ we have that $p_{-i}(v)\geq v$. The intuition is as follows\full{ (Proposition \ref{prop-less-than-mono} proves this formally)}: by definition, the monopolist price maximizes the revenue we can extract from a player, even ignoring all other players. Denote this revenue by $R_x$. If the value of the other player $j$ is $v_j<x$ and the potential payment is $p_{-i}(v_j)=x$ then player $i$ wins the item whenever $v_i\geq x$ and pays $x$ (because of Lemma \ref{lemma-sketch-main} and $v_i\geq x>v_j)$. Thus the revenue extracted from player $i$ if player $j$'s value is $v_j<x$ is also $R_x$. By definition of monopolist price, there is no value of $p_{-i}(v_j)$ that extracts higher revenue.

We now sketch the proof of Lemma \ref{lemma-sketch-main}. {For sake of concreteness we explicitly refer to the players here as player $1$ and player $2$.}  The proof is inductive, where the induction is over the values in the support of $\mathcal F$ from the highest to the monopolist price $x$. Let $v_1\geq x$ be the highest value such that for some player, without loss of generality player $2$, $p_{-2}(v_1)<v_1$, i.e., when the value of player $1$ is $v_1$ the potential payment of player $2$ is less than $v_1$ and for any $z>v_1$, $p_{-1}(z)\geq z$ and $p_{-2}(z)\geq z$. Let $v_2=p_{-2}(v_1)$ and consider the instance $(v_1,v_2)$. There are two possible cases, depending on whether $p_{-1}(v_2)\leq v_1$ or not. 

In both cases we will show that we can obtain a new mechanism by increasing $p_{-2}(v_1)$ by $\eps$, keeping the payment functions for higher values $z>v_1$ the same, and without decreasing the revenue comparing to the original mechanism. If $p_{-2}(v_1)+\eps \geq v_1$, then $v_1$ no longer violates the lemma in the new mechanism, so we can continue and consider the next value that violates Lemma \ref{lemma-sketch-main}. Else, we analyze the new mechanism in the instance $(v_1,v_2+\eps)$ according to the two cases, and conclude that we can increase $p_{-2}(v_1)$ by another $\eps$. We continue similarly until a mechanism with $p_{-2}(v_1)\geq v_1$ and then proceed to the next value that violates Lemma \ref{lemma-sketch-main}. At the end of the process we will have a mechanism with the property that for each player $i$ with value $v_i\geq x$ the potential payment of the other player is at least $v_i$. In addition, the revenue of the new mechanism is no smaller than the revenue of the original mechanism.

\paragraph*{Case I: $p_{-1}(v_2)\leq v_1$.} This case is depicted in Figure \ref{fig-smaller-price}.
		\begin{figure}[htb] 
	\begin{center}
		\begin{tikzpicture}[->,shorten >=1pt,auto,node distance=1.5cm, thin]
				\node at (0,2) {Player $1$};	
		\node at (3,2) {Player $2$};	
		\node (v1)  at (0,0.5) {$v_1$};
		\node (q1)  at (3,0.5) {$v_1$};
		\node (r1) at (3,1.5) {$.$};
		\node (l1) at (0,1.5) {$.$};
		\node (r2) at (3,1.2) {$.$};
		\node (l2) at (0,1.2) {$.$};
		\node (r3) at (3,0.9) {$.$};
		\node (l3) at (0,0.9) {$.$};				
		\node (l4) at (0,0.2) {$.$};				
				\node (l5) at (0,-0.1) {$.$};				
						\node (l6) at (0,-0.4) {$.$};				
								\node (l7) at (0,0.7) {$.$};				
		\node (z)  at (3,-0.2) {};
		\node (v2)  at (3,-1) {$v_2$};	
\draw (0,1.25) ellipse (0.3cm and 0.5cm);
\draw (3,1.25) ellipse (0.3cm and 0.5cm);	
		\draw (0.3,1.4) -- (2.75,1.4);
		\draw (2.7,1) -- (0.25,1);
			\draw (v1) -- (v2);
			\draw (v2) -- (l6);
		\end{tikzpicture}
		\caption{The setting of \full{Claim \ref{clm:smaller-price}}\itcs{Case I} illustrated for two players. The values in the support of each player are depicted as points from low to high. A directed edge from $v$ to $u$ implies that the potential payment that the player with value $v$ presents to the other player is $u$.
		} \label{fig-smaller-price}
	\end{center}
\end{figure}
We show that in this case setting $p_{-2}(v_1)=v_2+\eps$, i.e., increases the potential payment by $\eps$ and does not decrease the revenue\full{( the formal proof of this case is given in Claim \ref{clm:smaller-price})}. This is shown by fixing the value of player $1$ to be $v_1$ and comparing the expected revenue of $p_{-2}(v_1)=v_2$ and of $p_{-2}(v_1)=v_2+\eps$. Towards this end, we now analyze all instances of the form $(v_1,z)$ and show that the revenue does not decrease in each of them:

\begin{itemize}
\item When $z<v_2$ player $2$'s potential profit is negative in both cases so the allocation and payment remain the same with both payment functions.
\item When $z=v_2$ player $2$'s potential profit is non-positive in both cases. If player $1$'s potential profit is positive then player $1$ gets the item and pays the same in both cases. If player $1$'s potential profit is $0$, then when $p_{-2}(v_1)=v_2$ the potential profits of both players are $0$ and player $1$ will get the item by tie breaking, since his potential payment is higher. Thus, when $p_{-2}(v_1)=v_2+\eps$ player $2$'s potential profit becomes negative and the item is allocated to player $1$ for the same price.
\item When $v_1\geq z>v_2$, \full{Observation \ref{obs-small-is profitable} shows}\itcs{we show} that setting $p_{-2}(v_1)=v_2+\eps$ does not decrease the revenue. The intuition is as follows: for all values $z$ in the range player $2$'s potential profit is non-negative profit
 even when $p_{-2}(v_1)=v_2+\eps$, so the item will be allocated. Thus, our only concern is the following instances: the item was allocated to player $2$ when $p_{-2}(v_1)=v_2$ but is allocated to player $1$ when $p_{-2}(v_1)=v_2+\eps$. This can happen when the potential profit of player $1$ was equal to player $2$'s potential profit, but, due to the increase of $p_{-2}(v_1)$ the potential profit of player $1$ became strictly higher.
 The worry is that player $1$ pays less than $v_2$. We claim that if player $1$ becomes more profitable due to the increase in the potential payment of player $2$ then the payment of player $1$ is at least $v_2$. This is because if the potential profits are equal then $v_1-p_{-1}(z)=z-v_2$. However, we have that $v_1\geq z$ thus $p_{-1}(z)\geq v_2$.
\item $z>v_1$. In this case by our assumption we have that $p_{-1}(z)\geq z>v_1$. I.e., player $1$ is not profitable. The potential profit of player $2$ is non-negative when $p_{-2}(v_1)=v_2$ and when $p_{-2}(v_1)=v_2+\eps$ and thus player $2$ will receive the item. However, the revenue is strictly higher when $p_{-2}(v_1)=v_2+\eps$ in this case.
\end{itemize}

\paragraph*{Case II: $p_{-1}(v_2)>v_1$.} This case is depicted in Figure \ref{fig-higher-price}.
		\begin{figure}[htb] 
	\begin{center}
		\begin{tikzpicture}[->,shorten >=1pt,auto,node distance=1.5cm, thin]
				\node at (0,2) {Player $1$};	
		\node at (3,2) {Player $2$};	
		\node (v1)  at (0,0.5) {$v_1$};
		\node (q1)  at (3,0.5) {$v_1$};
		\node (r1) at (3,1.5) {$.$};
		\node (l1) at (0,1.5) {$.$};
		\node (r2) at (3,1.2) {$.$};
		\node (l2) at (0,1.2) {$.$};
		\node (r3) at (3,0.9) {$.$};
		\node (l3) at (0,0.9) {$.$};				
		\node (z)  at (3,-0.2) {};
		\node (v2)  at (3,-1) {$v_2$};	
\draw (0,1.25) ellipse (0.3cm and 0.5cm);
\draw (3,1.25) ellipse (0.3cm and 0.5cm);	
		\draw (0.3,1.4) -- (2.75,1.4);
		\draw (2.7,1) -- (0.25,1);
			\draw (v1) -- (v2);
			\draw (v2) -- (l3);
		\end{tikzpicture}
		\caption{The setting of \full{Claim \ref{clm:higher-price}}\itcs{Case II} illustrated for two players. The values in the support of each player are depicted as points from low to high. A directed edge from $v$ to $u$ implies that the potential payment that the player with value $v$ presents to the other player is $u$.
		} \label{fig-higher-price}
	\end{center}
\end{figure}
We again sketch the proof\full{and note that the formal proof can be found in Claim \ref{clm:higher-price}}. We first show that $p_{-1}(v_2)=v_1+\eps$. The proof relies on $p_{-2}(z)\geq z$ for all $z>v_1$ and on the regularity of $\mathcal F$, and is similar to our observation above that if Lemma \ref{lemma-sketch-main} holds then for $z<x$ we have that $p_{-1}(z)=x$.

We consider obtaining a new mechanism by applying the following simultaneous manipulations on the payment functions: $p_{-2}(v_1)=v_2+\eps$ and $p_{-1}(v_2)=v_1$. In other words, we increase $p_{-2}(v_1)$ by $\eps$ and decrease $p_{-1}(v_2)$ by $\eps$. On one hand, the changes increase the revenue by $\eps$ in instances $(v_1,z)$, when $z>v_2$ (the analysis is similar to the previous case). On the other hand, the changes decrease the revenue by $\eps$ in instances $(z,v_2)$, $z>v_1$. When analyzing the conditions under which making both changes does not decrease the revenue the requirement that the distribution is standard organically emerges.

This concludes the sketch of proof of Lemma \ref{lemma-sketch-main} and of Theorem \ref{thm-standard} for the case of two players.

\section{Correlated Distributions: Moral Beats Truthfulness}\label{sec-correlated}

We now prove that moral mechanisms sometimes generate more revenue than their truthful counterparts. For this result to hold we need the values of the players to be correlated. The proof relies on the following theorem:

\begin{theorem}[essentially \cite{DU18}]\label{thm-correlated-truthful}
	For every small enough $\delta, \eps>0$, and $1\geq \alpha \geq \eps$, there exists a distribution $\mathcal H$ over the values of two players with the following properties:
	\begin{enumerate}
		\item The support of $\mathcal H$ is finite.
		\item All values in the support of $\mathcal H$ are at least $1$.
		\item\label{item-distance} For each two possible values of player $1$ $v_1,v'_1$ in the support of $\mathcal H$ we have that $|v_1-v'_1|>\eps$.
		\item For each possible value $v_2$ of player $2$ in the support of $\mathcal H$ we have that $1\leq v_2\leq 1+ \eps\cdot \alpha$.
		\item The revenue that every truthful mechanism generates on $\mathcal H$ is at most $\frac e {e-1}+\delta$. 
		\item There exists a revenue-maximizing truthful mechanism $M$ in which player $2$ is never allocated (i.e., all revenue is generated from player $1$). In this mechanism, the item is not allocated at all with probability (over $\mathcal H$) at least $\frac 1 {e-1}$.
	\end{enumerate}
\end{theorem}

We now provide some intuition for the proof of Theorem \ref{thm-correlated-truthful}. Let $\mathcal H'$ be a distribution that is identical to $\mathcal H$ except that $v_2$ is always $1$. The optimal mechanism for $\mathcal H'$ is to make a take-it-or-leave it offer to player $1$, and if player 1 declines, sell the item to player $2$ at price $1$. However, in the distribution $\mathcal H$ the low-order bits of $v_2$ determine the distribution that $v_1$ is drawn from. Ideally, the mechanism would query the low order bits of $v_2$, determine the distribution that $v_1$ is drawn from, make a take-it-or-leave-it-offer to player $1$ and if player $1$ declines sell the item to player $2$ and charge $1$. However, this will contradict truthfulness since player $2$ can possibly misreport his low-order bits in order to fool the mechanism to believe that $v_1$ is drawn from another distribution, this will cause the mechanism to make a take-it-or-leave-it offer which is too high, player $1$ will decline and player $2$ will incur positive profit from taking the item and paying $1$. 

While this is only one way to fool a specific algorithm, \cite{DU18} shows that every truthful mechanism can be similarly fooled. Thus, intuitively, a truthful mechanism can either (1) sell the item at a fixed price of $1$ (2) make a take-it-or-leave-it offer to player $1$ without querying the low-order bits of $v_2$ and if player $1$ declines then the item is sold to player $2$ at price $1$, or (3) query the low order bits of $v_2$, make a take-it-or-leave-it offer to player $1$ and do not sell the item if he declines. The last auction is in fact the revenue-maximizing truthful mechanism.

We prove that moral mechanisms can generate more revenue on $\mathcal H$. For simplicity, consider $\alpha=1$ and the following profit maximizer: player $1$'s potential payment is (almost) the best take-it-or-leave-it offer that can be made after querying player $2$'s low-order bits. Player $2$'s potential payment is $1$. We show that player $1$ receives the item if his profit is positive, and if his profit is negative then player $2$ gets the item and pays $1$. By our discussion above, this mechanism generates more revenue than the optimal truthful mechanism for $\mathcal H$. Formally:

\begin{theorem}\label{thm-correlated-moral}
	Fix $1\geq \alpha>0$ and $\eps>0$. There exists a distribution $\mathcal H$ for which an $\alpha$-moral mechanism  generates a revenue of $\frac e {e-1}+\delta + \frac{1}{e-1} - \frac{e+1}{e-1}\cdot \eps -\eps\cdot \delta$ (which is $\approx \frac {e+1} {e-1}$) on $\mathcal H$ whereas the revenue of every truthful mechanism is at most $\frac e {e-1}+\delta$.
	%
	%
\end{theorem}
\begin{proof}
Let $\mathcal H,\eps, \delta, \alpha, M$ be as in Theorem \ref{thm-correlated-truthful}. We define the following profit maximizer $M'$: the prices that player $2$ presents to player $1$ remain the same, except that we decrease them by $\eps$. Player $1$ always presents a price of $1$ to player $2$.
	
	We now analyze the revenue of $M'$. We claim that in every instance in the support of $\mathcal H$ where $M$ allocates the item to player $1$ for a payment of $q$, $M'$ will also allocate the item to player $1$ for a payment of $q-\eps$. Moreover, when $M$ does not allocate the item at all, $M'$ allocates the item and charges player $2$ a payment of $1$. In other words, when $M$ allocates the item to player $1$ $M'$ does so as well and the revenue barely changes. When $M$ does not allocate the item (and recall that this happens with probability at least $\frac 1 {e-1}$), $M'$ generates a revenue $1$. This will imply the theorem.
	
	To prove the above claim, recall that $M$ is truthful and thus is a profit maximizer in which at most one player has a positive profit. Let $(v_1,v_2)$ be an instance in the support of $\mathcal H$ in which player $1$ is allocated and is charged $q$. That is, $v_1-q\geq 0$. Consider this instance in $M'$. The potential profit of player $1$ increases now to $v_1-(q-\eps) \geq \eps$. The potential profit of player $2$ is 
	$v_2-1\leq 1+ \eps \cdot \alpha-1=\eps \cdot \alpha$. We get that $M'$ allocates the item to player $1$ (if player $2$ misreports his gain will be at most $\eps \cdot \alpha$ and the loss of player $1$ from this misreport is at least $\eps$) and charges $\eps$ less, compared to the revenue of $M$. Thus, the aggregate loss of $M'$ with respect to $M$ is at most $\eps\cdot (\frac e {e-1}+\delta)$.  

	Now consider an instance $(v_1,v_2)$ in the support of $\mathcal H$ where the item is not allocated at all. Note that player $2$ is profitable in every instance in the support of $\mathcal H$, since $v_2\geq 1$ and his potential payment is $1$, so $M'$ does allocate the item. If the item is allocated to player $2$, then the payment is $1$ and we are done.
	 We claim that the revenue is high also when the item is allocated to player $1$, despite decreasing his potential payment by $\eps$. This follows because $M$ is revenue maximizing, and thus the payments that player $1$ pays must be in the support of $\mathcal H$ (otherwise, the payments can be increased and the revenue of $M$ increases as well). This implies that all prices presented by $M$ are at least $1$, and thus the prices presented by $M'$ are at least $1-\eps$, as needed. The aggregate gain of $M'$ with respect to $M$ is at least $\frac{1}{e-1}\cdot (1-\eps)$. Putting this together and rearranging, we get that the total gain of $M'$ with respect to $M$ is at least $\frac{1}{e-1} - \eps-\frac{e+1}{e-1}\eps\cdot \delta $. 
\end{proof}

Theorem \ref{thm-correlated-moral} shows that the multiplicative gap between the revenue of an optimal truthful mechanism and an optimal moral mechanism can be as large as $\frac {e+1} e$. We now shows that this gap cannot be larger than $2$. We obtain this result by showing that the lookahead auction of Ronen \cite{ronen-lookahead} provides a $2$ approximation not just with respect to the optimal truthful mechanism as shown in \cite{ronen-lookahead} but also with respect to the optimal $\alpha$-moral mechanism. 

\begin{theorem} \label{thm-cor-exp}
Fix a distribution $\mathcal H$ and an $\alpha$-moral mechanism $M$, $1\geq \alpha\geq 0$. Denote by $rev(M)$ the expected revenue of $M$ over $\mathcal H$. There is a truthful mechanism $M'$ whose revenue is at least $\frac {rev(M)} 2$.
\end{theorem}
\begin{proof}
		The mechanism $M'$ is Ronen's lookahead auction: without loss of generality order the players so that $v_1\geq v_2\geq \ldots \geq v_n$. Compute the marginal distribution $\mathcal H_1$ of $v_1$ given that $v_1\geq v_2$. Make a revenue maximizing offer to player $1$ given that his value is drawn from $\mathcal H_1$. If player $1$ accepts then he pays the take-it-or-leave-it price, otherwise the item is not sold at all. It is not hard to see (and is proved in \cite{ronen-lookahead}) that the lookahead mechanism is truthful.
		
		We will now prove that the revenue of the lookahead auction is at least $\frac {rev(M)} 2$. In fact, we may assume that $M$ is a profit maximizer (i.e., $\alpha=1$), this only makes our result stronger.
		
		There are two possible cases. In the first case, at least half of the revenue that $M$ generates is made in instances where the item is allocated to the highest value player. Here we observe that the potential profit of the highest player is non-negative. Thus, by definition, we have that the best take-it-or-leave-it offer to the highest player will generate in expectation more revenue than collecting the potential payment of player $1$ (as in $M$), which shows that in this case, the lookahead auction provides a $2$-approximation.
		
		Else, more than half of the revenue that $M$ generates is made in instances where the item is allocated to a player that does not have the highest value. Note that the expected revenue of $M$ can be bounded in this case by twice the expected second-highest value. 
		Since the take-it-or-leave-it price can be the second-highest value, the item is sold in every instance, and thus the revenue is the expectation of the second-highest value. This gives us a $2$-approximation in this case as well.
\end{proof}


\bibliography{deception}

\full{
\appendix
\section{Truthful Revenue Maximization for Independent Distributions}\label{sec-truthful}

In this section we consider optimal (revenue-maximizing) truthful mechanisms for $n$ players with discrete valuations. We limit our attention to distributions that satisfy Chung and Ely's \cite{chung2007foundations} discrete analogue of Myerson's regularity. 

Denote by $x_i$ the monopolist price of distribution $\mathcal F_i$, i.e., the maximal value $v$ such that $\vv_i(v)\leq 0$.
For a vector $v_{-i}$ that specifies the values of all players except player $i$, we let $q_{-i}(v_{-i})=\max(x_i,\vv_i^{-1}(\hv))$, where $\hv=\max_{j} (\vv_j(v_j))$ is the highest virtual value that corresponds to a value in $v_{-i}$. We characterize the optimal auction in terms of its payment functions. We let $p_{-i}(v_{-i})$ denote the price (i.e., critical value) that is presented to player $i$ when the values of the other players are according to $v_{-i}$. 

The characterization we obtain is slightly different than Myerson's characterization for continuous and atomless distribution. For example, in the case of two players with values sampled independently from the same regular (continuous and atomless) distribution, the optimal auction is a second price auction. In our discrete setting, the optimal auction is that the critical price of player $1$ is $v_2$, and the critical price of player $2$ is $v_1+\eps$.

\begin{theorem}
Let $M$ be a truthful mechanism for selling one item to $n$ players. Suppose that the value $v_i$ of player $i$ is drawn independently from a regular distribution $\mathcal F_i$. Let $M$ be an optimal mechanism with payment functions $p_{-1}, \ldots, p_{-n}$. Consider an instance $(v_1,\ldots, v_n)$. For each player $i$, if $q_{-i}(v_{-i})=x_i$ then $p_{-i}(v_{-i})=x_i$. Among the players with $q_{-i}(v_{-i})>x_i$, it holds there is exactly one player such that $p_{-i}(v_{-i})=q_{-i}(v_{-i})$ and for every other player $j\neq i$, $p_{-j}(v_{-j})=q_{-j}(v_{-j})+\eps$.
\end{theorem}

We now prove the theorem. For simplicity we assume that all the values in the support of each $\mathcal F_i$ are identical and we denote this set of values by $V$. We define a partial order $\succ$ over all vectors in $V^{n-1}$. For $v \in  V^{n-1}$, let $t_v$ be the vector that its $i$'th coordinate is $\vv_i(v_i)$. $sort(v)$ is the $n-1$ entries of the vector $t_v$ sorted from largest to smallest. For $v, v' \in  V^{n-1}$, $v \succ v'$ if $sort(v)$ is lexicographically bigger than $sort(v')$. For example, $v \succ v'$ if the highest virtual value that corresponds to a value in $v$ is greater than the highest virtual value that corresponds to a value in $v'$.

\begin{proposition}\label{prop-main-truthful}
Fix some mechanism $M$ with payment functions $p_{-1},\ldots, p_{-n}$. Suppose that for some $v_{-i}\in V^{n-1}$ with $\vv_i^{-1}(\hv)\geq x_i$ we have that for every $\hv_{-j}\succ v_{-i}$ it holds that $p_{-j}(\hv_{-j})\geq q_{-j}(\hv_{-j})$. Then, there exists a mechanism $M'$ with payment functions $p'_{-1},\ldots, p'_{-n}$ such that for every $\hv_{-j}\succ v_{-i}$ it holds that $p'_{-j}(\hv_{-j})\geq q_{-j}(\hv_{-j})$. In addition, $p'_{-i}(v_{-i}) \geq q_{-i}(v_{-i})$. Moreover, the revenue of $M'$ is not lower than the revenue of $M$.
\end{proposition}

Before proving the proposition we show that the proposition immediately implies a characterization of the payment functions for all values above the monopolist prices:

\begin{corollary}
There exists a revenue maximizing mechanism $M$ with payment functions $p_1,\ldots, p_n$ such that for every $v_{-i}$ with $\vv^{-1}(\hv)\geq x_i$ we have that $p_{-i}(v_{-i})\geq q_{-i}(v_{-i})$.
\end{corollary}
\begin{proof}
If the claim does not hold, for some $v_{-i}$ that satisfies the condition of the proposition we have that $p_{-i}(v_{-i})<  q_{-i}(v_{-i})$. If there are several such $v_{-i}$'s, we assume that $v_{-i}$ is a maximal such vector according to $\succ$. Note that by the maximality of $v_{-i}$, we have that for every $\hv_{-j}\succ v_{-i}$, $p_{-j}(\hv_{-j})\geq q_{-i}(\hv_{-j})$. The conditions of the proposition thus hold and thus there is another mechanism $M'$ with payment functions $p'_1,\ldots, p'_n$ such that $p'_{-i}(v_{-i}) \geq q_{-i}(v_{-i})$. We continue similarly to the next vector $v'_{-i}$ that violates the corollary, fixing it by obtaining a new mechanism, and so on.
\end{proof}


We are now ready to prove Proposition \ref{prop-main-truthful}:

\begin{proof}
Let $v_{-i}$ be such that the conditions of the proposition hold but $p_{-i}(v_{-i})<q_{-i}(v_{-i})$. Let $v_i=p_{-i}(v_{-i})$. 

\begin{claim} \label{clm:smaller-price-truthful}
Fix some mechanism $M$ with payment functions $p_{-1},\ldots,p_{-n}$. Let $v_{-i}$ be such that the conditions of Proposition \ref{prop-main-truthful} hold in $M$ and $p_{-i}(v_{-i})<q_{-i}(v_{-i})$. Let $v_i=p_{-i}(v_{-i})$ and suppose that in the instance $(v_i,v_{-i})$ there exists some player $j\neq i$ such that $v_j\geq p_{-j}(v_{-j})$. Then, there exists a mechanism $M'$ with payment functions $p'_{-1}, \ldots, p'_{-n}$ that generates at least as much revenue as $M$. Furthermore, all the conditions of Proposition \ref{prop-main-truthful} hold for $v_{-i}$ in $M'$, but in the instance $(v_i,v_{-i})$ it holds that for every player $j\neq i$, $v_j< p'_{-j}(v_{-j})$.
\end{claim}

\begin{claim} \label{clm:higher-price-truthful}
Fix some mechanism $M$ with payment functions $p_{-1},\ldots,p_{-n}$. Let $v_{-i}$ be such that the conditions of Proposition \ref{prop-main-truthful} hold in $M$ and $p_{-i}(v_{-i})<q_{-i}(v_{-i})$. Let $v_i=p_{-i}(v_{-i})$ and suppose that in the instance $(v_i,v_{-i})$, for every $k\neq i$, $p_{-k}(v_{-k})>v_k$. Let $M'$ be a mechanism with payment functions $p'_{-1}, \ldots, p'_{-n}$ which are identical to the payment functions $p_{-1},\ldots, ,p_{-n}$ of $M$ except that $p'_{-i}(v_{-i})= p_{-i}(v_{-i})+\eps$ and $p'_{-j}(v_{-j})=v_j$, where $j$ is the player with the highest virtual value in $v_{-i}$. Then, the revenue of $M'$ is not lower than the revenue of $M$. 
\end{claim}
Proofs of these claims can be found in Subsections \ref{subsec-proof-smaller-truthful} and \ref{subsec-proof-higher-truthful}. As in the outline, the two claims imply the proposition: find the first vector $v_{-i}$ that violates the proposition, use one of the claims to obtain another mechanism where $p_{-i}(v_{-i})$ is increased by $\eps$, and repeat until $p_{-i}(v_{-i})\geq q_{-i}(v_{-i})$. Then continue similarly to fix the next vector $v'_{-i'}$ that its price is too small.
\end{proof}

We now characterize the payment functions for vectors $v_{-i}$ such that $\vv_i^{-1}(\hv)< x_i$, i.e., the highest virtual value in $v_{-i}$ is negative. 

\begin{proposition}\label{prop-less-than-mono-truthful}
Consider a revenue maximizing mechanism $M$ with payment functions $p_{-1},\ldots, ,p_{-n}$ where values are sampled independently from regular distributions $\mathcal F_1,\ldots, \mathcal F_n$. Denote by $x_i$ the monopolist price of $\mathcal F_i$. If for every $v_{-i}$ such that $\vv_i^{-1}(\hv)\geq x_i$ it holds that $p_{-i}(v_{-i})\geq q_{-i}(v_{-i})$, then for every $v_{-i}$ such that $q_{-i}(v_{-i})= x_i$ it holds that $p_{-i}(v_{-i})=x_i$.
\end{proposition}
\begin{proof}
If our goal is only to maximize the revenue extracted from player $i$ (ignoring the contribution of the other players) we would set the threshold price of player $i$ to $x_i$, as $x_i\in\arg\max_p p(1-F_i(p-\eps))$. 

Now observe that when $p_{-i}(v_{-i})= x_i$  the revenue of the mechanism (conditioned on the values of all players except player $i$ are according to $v_{-i})$ is indeed $\max_p(p(1-F_i(p-\eps)))$, all of it due to payments made by player $i$. To see this, let $v^z_{-j}$ be the vector obtained from $(z, v_{-i})$ by removing the $j$'th coordinate, for $z\geq x_i$. Since $q_{-i}(v^z_{-j})> x_i$, for every $j\neq i$ (by the conditions of the proposition), we have that only player $i$ is profitable in those instances. Thus, when $p_{-i}(v_{-i})= x_i$ we extract as much revenue as possible from player $i$ while preserving feasibility. This proves that setting $p_{-i}(v_{-i})=x_i$ generates more revenue than setting $p_{-i}(v_{-i})=\bar x$, for $\bar x \geq x_i$.  Note that if we set $p_{-i}(v_{-i})=\underline x$, for $\underline x_i< x_i$, we extract less revenue from player $i$ and may hurt feasibility as there might be more than one player with strictly positive profit. We conclude that setting $p_{-i}(v_{-i})= x_i$ maximizes the revenue.
%
%
%
%
%
%
%
 \end{proof}

Up until now we have established that for every player $i$ and $v_{-i}$, $p_{-i}(v_{-i})\geq q_{-i}(v_{-i})$. To provide the exact form of the payment functions we will need the following lemma:

\begin{lemma} \label{lem-rev-max-truthful}
Consider a revenue maximizing mechanism $M$ with payment functions $p_{-1},\ldots, ,p_{-n}$ where values are independently sampled from regular distributions $\mathcal F_1,\ldots, \mathcal F_n$. Denote by $x_i$ the monopolist price of $\mathcal F_i$. Suppose that for some $v_{-i}$ and $t\geq x_i$ it holds that all players except player $i$ have negative profit in every instance $(z,v_{-i})$, $z\geq t$. Then, $p_{-i}(v_{-i})\leq t$.
%
%
%
\end{lemma}
\begin{proof}
Let $F_i, f_i$ be the c.d.f. and p.d.f. of $\mathcal F_i$, respectively. Towards contradiction, let $v_{-i}$, $t$ be such that all the conditions hold but $p_{-i}(v_{-i})>t$. We will show that the mechanism $M'$ with identical payment functions $p'_{-1}, \ldots ,p'_{-n}$ except that $p'_{-i}(v_{-i})=p_{-i}(v_{-i})-\eps$ generates more revenue.
	 
The only difference in the revenue of $M$ and $M'$ is in instances $(z,v_{-i})$, $z\geq t$. In these instances, by the conditions of the lemma, only player $i$ might be profitable in the instance $(z,v_{-i})$ in both mechanisms. Thus, we only need to compare the difference in the revenue in these instances. Denote $v_i=p_{-i}(v_{-i})$. Let $Rev^M(z,v_{-i})$ be the revenue of $M$ in the instance $(z,v_{-i})$.

	\begin{itemize}	
		\item  The revenue of $M$ in these instances is: $$\int_t^1 f_i(z) \cdot Rev^M(z,v_{-i})dz =v_i\cdot (1-F_i(v_i-\eps))$$

		\item The revenue of $M'$ in these instances is: $$\int_t^1 f_i(z) \cdot Rev^{M'}(z,v_{-i})dz =(v_i-\eps)\cdot (1-F_i(v_i-\eps)+f_i(v_i-\eps))$$

%
	\end{itemize}

Therefore, the revenue of $M'$ is at least that of $M$ if: 
\begin{align*}
(v_i-\eps)\cdot (1-F_i(v_i-\eps)+f_i(v_i-\eps)) &\geq v_i\cdot (1-F_i(v_i-\eps))\\
-\eps (1-F_i(v_i-\eps))+(v_i-\eps)f_i(v_i-\eps) &\geq 0 \\
v_i-\eps - \frac {\eps \cdot (1-F_i(v_i-\eps))} {f_i(v_i-\eps)} &\geq 0\\
\end{align*}

I.e., $\vv(v_i-\eps)\geq 0$, which is true since $\mathcal F_i$ is a regular distribution and $v_i> x_i$.
%
%
%
\end{proof}
	
We can now complete the proof:

\begin{claim}\label{claim-truthful-characterization}
Let $M$ be an optimal mechanism with payment functions $p_{-1}, \ldots, p_{-n}$. Suppose that for every player $i$ and $v_{-i}$, $p_{-i}(v_{-i})\geq q_{-i}(v_{-i})$. Consider an instance $(v_1,\ldots, v_n)$. For every $i$ such that $q_{-i}(v_{-i})=x_i$ we have that $p_{-i}(v_{-i})=x_i$. For the other players with $q_{-i}(v_{-i})>x_i$, it holds there is exactly one such player such that $p_{-i}(v_{-i})=q_{-i}(v_{-i})$ and for every other such player $j\neq i$, $p_{-j}(v_{-j})=q_{-j}(v_{-j})+\eps$.
\end{claim}
\begin{proof}
Consider some instance $(v_1,\ldots, v_n)$. If there is one player $i$ with $q_{-i}(v_{-i})=x_i$ then $p_{-i}(v_{-i})=x_i$, by Proposition \ref{prop-less-than-mono-truthful}.

From now on we limit our attention to the other players. We start by showing that for at least one player $i$, $p_{-i}(v_{-i})=q_{-i}(v_{-i})$. Suppose that for every player $j$, $p_{-j}(v_{-j})>q_{-j}(v_{-j})$. Fix some player $i$ and let $v^z_{-k}$ the vector obtained from $(z,v_{-i})$ by removing the $k$'th coordinate. Observe that by the conditions of the claim and monotonicity of $q_{-k}$'s, for every $k\neq i$ it holds that $p_{-i}(v^z_{-k})>q_{-i}(v_{-i})$. Thus by Lemma \ref{lem-rev-max}, $p_{-i}(v_{-i})=q_{-i}(v_{-i})$.

By a similar argument, for every player $i$ and $v_{-i}$, $p_{-i}(v_{-i})\in \{q_{-i}(v_{-i}),q_{-i}(v_{-i})+\eps\}$. Suppose that there are at least two players $i,j$ such that $p_{-i}(v_{-i})=q_{-i}(v_{-i})$ and $p_{-j}(v_{-j})=q_{-j}(v_{-j})$. Notice that the profit of all such players is $0$ since by definition of the $q_{-i}$'s the players have the same virtual value. Without loss of generality, suppose that $j$ is such a player with a maximal value. We may assume that the mechanism allocates player $j$ the item and charges $v_j$. 

Consider $p_{-i}(v_{-i})=q_{-i}(v_{-i})+\eps$. Note that this change does not affect the feasibility, as it is always feasible to increase the critical values. It suffices to analyze the expected change in revenue given that the values of all players except $i$ are according to $v_{-i}$. For $z<v_i$ player $i$ is not profitable before and after the change. The revenue in the instance $(v_1,\ldots, v_n)$ remains the same since the profit of player $i$ is negative and player $j$ still has a non-negative profit. When $z>v_i$, player $i$ either has a zero profit (when his value is $v_i+\eps$, we can allocate to either $i$ or $j$, depending on whether $v_i+\eps\geq v_j$) or a positive profit. Note that if $i$ has positive profit after the change he had positive profit before the change. In that case, the payment of $i$ has increased by $\eps$. We conclude that after the change the mechanism generates in every instance no less revenue. 
\end{proof}

\subsection{Proof of Claim \ref{clm:smaller-price-truthful}}\label{subsec-proof-smaller-truthful}
The setting of the claim is illustrated in Figure \ref{fig-smaller-price}. We analyze the difference in the revenue of $M$ and in the revenue of $M'$. Suppose first that there is some player $j$ with positive profit in the instance $(v_i, v_{-i})$ (i.e., $v_j > p_{-j}(v_{-j})$). Let $M'$ be the mechanism with payment functions $p'_{-1}, \ldots, p'_{-n}$ which are identical to the payment functions $p_{-1},\ldots, ,p_{-n}$ of $M$ except that $p'_{-i}(v_{-i})= p_{-i}(v_{-i})+\eps=v_i+\eps$. Note that since the payment functions of the mechanisms are almost identical, it suffices to fix the values of all players except player $i$ according to $v_{-i}$ and analyze the difference in the expected revenue, where expectation is taken over the value of player $i$. 

Consider instances of the form $(z,v_{-i})$. If $z<v_i$, $z$ is less than the payment of $i$ in both $M$ and $M'$. Thus, the revenue of the mechanism depends only on the other payment functions which have not changed and is therefore identical. In the instance $(v_i,v_{-i})$ player $j$'s profit is positive in both $M$ and $M'$, thus the allocation and payment are the same. When $z>v_i$, we observe that player $i$ has a positive profit in $M$ and thus he is the player that is allocated. Hence, in every such instance $(z,v_{-i})$ all players except player $i$ have non-positive profit. Note that in all these instances player $i$ has either zero profit (when $z=v_i+\eps$) or positive profit (when $z>v_i+\eps$) since $p'_{-i}(v_{-i})=v_i+\eps$. We may therefore assume that player $i$ receives the item and pays $\eps$ more. We conclude that in this case the expected revenue of $M'$ is bigger than that of $M$.

The other case that is when there is some player $j$ ($j\neq i$) whose value is his critical price (i.e., $v_j = p_{-j}(v_{-j})$). If there are several such players (i.e., multiple players that their critical price equal their value), we let $j$ be the player with the highest value among those. The analysis of this case is similar to the previous case. We first consider the case where $v_j>v_i$. Note that to maximize revenue, $M$ will allocate player $j$ the item (and charge $v_j$) in the instance $(v_i,v_{-i})$.


Let $M'$ be the mechanism that was defined above. By tie-breaking (if there are several players with zero profit), we may assume that $M'$ allocates the item to player $j$ and charges $v_j$ in the instance $(v_i,v_{-i})$, so the revenue of $M$ and $M'$ is the same in the instance $(v_i,v_{-i})$. 

Consider an instance $(z,v_{-i})$. If $z<v_i$, $z$ is less than $i$'s critical price (i.e., payment) in both $M$ and $M'$, the revenue of the mechanism depends only on the other payment functions which have not changed and thus is identical. When $z>v_i$, we observe that player $i$ has a positive profit in $M$ and thus gets the item. Thus, in every such instance $(z,v_{-i})$ all players except player $i$ have non-positive profit. Note that in all these instances player $i$ has either zero profit (when $z=v_i+\eps$) or positive profit (when $z>v_i+\eps$) since $p'_{-i}(v_{-i})=v_i+\eps$. Thus we may assume that player $i$ receives the item and pays $\eps$ more. The expected revenue of $M'$ is bigger than that of $M$.

When $v_j\leq v_i$, the proof is almost identical except that we define $M'$ to have the same payment functions as $M$ except $p'_{-j}(v_{-j})= p_{-j}(v_{-j})+\eps$, fix $v_{-j}$ and consider instances of the form $(z,v_{-j})$.

\subsection{Proof of Claim \ref{clm:higher-price-truthful}}\label{subsec-proof-higher-truthful}


The setting of the claim is illustrated in Figure \ref{fig-higher-price}. By the assumptions of the claim, $p_{-j}(v_{-j})>v_j$. We first claim that $p_{-j}(v_{-j})= v_j+\eps$. To prove this, observe that for every $z>v_j$ all players $k\neq j$ have negative profit, because for every player $k\neq j$ it holds that $v^z_{-k} \succ v_{-i}$, where $v^z_{-k}$ denotes the vector obtained from $(z, v_{-i})$ by removing the $k$'th coordinate. Hence by the conditions of Proposition \ref{prop-main-truthful} we have that $p_{-k}(v^z_{-k})> q_{-k}(v^z_{-k})> q_{-k}(v_{-k}) \geq  v_k$. We can apply Lemma \ref{lem-rev-max-truthful} and conclude that $p_{-j}(v_{-j})= v_j+\eps$. The heart of the proof is the following lemma:

	\begin{lemma} \label{lem-sim-change-truthful}
		Let $M$ be a mechanism for players that their values are independently sampled from regular distributions $\mathcal F_1,\ldots, \mathcal F_n$ with payment functions $p_{-1},\ldots, p_{-n}$. Consider an instance $(v_1,\ldots, v_n)$. Let $j$ be the player with the highest virtual value in $v_{-i}$. Suppose that:
		\begin{enumerate}
			\item $p_{-i}(v_{-i})=v_i<q_{-i}(v_{-i})$.
			\item $p_{-j}(v_{-j})=v_j+\eps$.
			\item For every $z>v_j$, the only player that might have a non-negative profit in the instance $(z,v_{-j})$ is player $j$.
		\end{enumerate}
		Then, the mechanism $M'$ which has payment functions $p'_{-1},\ldots,  p'_{-n}$ that are identical to $p_{-1}, \ldots, p_{-n}$ except that $p'_{-i}(v_{-i})=v_i+\eps$ and $p'_{-j}(v_{-j})=v_j$ generates at least as much revenue as $M$.
	\end{lemma}

Before proving the lemma notice that all of its conditions hold. Thus, by Lemma \ref{lem-sim-change-truthful} we have that the revenue of $M'$ is at least as large as that of $M$.

\begin{proof}(of Lemma \ref{lem-sim-change-truthful})
For each distribution $\mathcal F_k$, let $F_k,f_k$ be the c.d.f. and p.d.f of $\mathcal F_k$, respectively. The only instances that the revenue of $M$ and $M'$ might differ in are those in which the values of the players are according to $v_{-i}$ or according to $v_{-j}$. The expected revenue of $M$ in those instances is:
	\begin{align} \label{eq-swap-p-truthful}
	\prod_{k\neq i}f_k(v_{k})&\int_{0}^{1} \delta(z\neq v_i) \cdot f_i(z) \cdot  Rev^M(z,v_{-i})dz + \\
	&\prod_{k\neq j}f_k(v_{k})\int_{0}^{1} \delta(z\neq v_j) \cdot f_j(z) \cdot Rev^M(z,v_{-j})dz + \prod_{k}f_k(v_{k})\cdot v_i\nonumber
	\end{align}
	where by $Rev^M(z,v_{-i}) $ we mean the revenue of $M$ in the instance $(z,v_{-i})$. Recall that $\delta(\cdot)$ is an indicator function that gets a boolean condition and returns $1$ if the condition is true and $0$ otherwise. Similarly, the expected revenue of $M'$ is:
	\begin{align}  \label{eq-swap-p'-truthful}
		\prod_{k\neq i}f_k(v_{k})&\int_{0}^{1} \delta(z\neq v_i) \cdot f_i(z) \cdot  Rev^{M'}(z,v_{-i})dz +\\ 
		&\prod_{k\neq j}f_k(v_{k})\int_{0}^{1} \delta(z\neq v_j) \cdot f_j(z) \cdot Rev^{M'}(z,v_{-j})dz + \prod_{k} f_k(v_{k})\cdot v_j\nonumber
	\end{align}
	We would like to show that the expected revenue of $M'$ is at least as high as that of $M$  ($(\ref{eq-swap-p-truthful}) \leq (\ref{eq-swap-p'-truthful})$). By rearranging we get that it is sufficient to show that:
	\begin{align} \label{eq-3-truthful}
	f_j(v_{j})&\int_{0}^{1} \delta(z\neq v_i) \cdot f_i(z) \cdot \Big(  Rev^M(z,v_{-i})-Rev^{M'}(z,v_{-i}) \Big)dz  \leq \\  
	&f_i(v_{i})\int_{0}^{1} \delta(z\neq v_j) \cdot f_j(z) \cdot \Big( Rev^{M'}(z,v_{-j}) - Rev^M(z,v_{-j})\Big)dz 
	+ f_j(v_j)\cdot f_i(v_i) \cdot (v_j-v_i) \nonumber
	\end{align}
		Notice that	in every instance ($z, v_{-i})$, $z>v_i$, only player $i$ is allocated in $M$ since his profit is positive. In these instances, the revenue of $M'$ will be bigger by $\eps$ than that of $M$. When $z<v_i$ player $i$ is not profitable in both mechanisms and thus the revenue of $M$ and $M'$ is the same. It therefore holds that: 
	\begin{align*}
	f_j(v_{j})\int_{0}^{1} \delta(z\neq v_i)  \cdot f_i(z) \cdot  \Big(  Rev^M(z,v_{-i})-Rev^{M'}(z,v_{-i}) \Big)dz 
	=	f_j(v_j) \cdot \eps \cdot (1-F_i(v_i))
	\end{align*}
	We now consider instances of the form $(z,v_{-j})$. For every $z> v_j$ we have that the only profitable player in these instances is player $j$. Since $p_{-j}(v_{-j})-p'_{-j}(v_{-j})=\eps$, $M'$ loses an extra revenue of $\eps$ comparing to $M$. When $z<v_j$, player $j$ is not profitable in both mechanisms and thus their revenue is the same. We therefore have that:
	\begin{align*}
	f_i(v_i)\int_{0}^{1} \delta(z\neq v_i) \Big( Rev^{M'}(z,v_{-j}) - Rev^M(z,v_{-j})\Big)dz   =  f_i(v_i) \cdot \eps  \cdot(1-F_j(v_j))
	\end{align*}
	Putting this together, we get that  a sufficient condition for inequality (\ref{eq-3-truthful}) to hold is:
	\begin{align*}
	f_j(v_j) \cdot \eps \cdot (1-F_i(v_i)) \leq f_i(v_i) \cdot \eps  \cdot(1-F_j(v_j)) + f_j(v_j)\cdot f_i(v_i) \cdot (v_j-v_i)
	\end{align*}
	By dividing by $ f_i(v_i) \cdot  f_j(v_j)$, and rearranging we get that:
	\begin{align*}
	v_j-\eps\cdot \dfrac{1-F_j(v_j)}{f_j(v_j)}  \geq v_i-\eps\cdot \dfrac{1-F_i(v_i)}{f_i(v_i)}
	\end{align*}
	I.e., a sufficient condition for inequality (\ref{eq-3-truthful}) is  $\vv_j(v_j) \geq \vv_i(v_i)$, which holds by assumption.
\end{proof}

\section{Missing Proofs from Section \ref{sec-properties}} \label{app-sec-3}

\begin{proposition} \label{prop-lattice}
	Consider an allocation function $f$ for two players. Let $M_1=(f,p_1)$ and $M_2=(f,p_2)$ be two $1$-moral mechanisms that implement the same allocation function but with different payment functions. Define new payment functions $p_{max}$ and $p_{min}$ as follows:
	\begin{align*}
	p_{max}(v_1,\ldots, v_n)_i &= \max\{p_1(v_1,\ldots, v_n)_i, p_2(v_1,\ldots, v_n)_i\} \\
	p_{min}(v_1,\ldots, v_n)_i &= \min\{p_1(v_1,\ldots, v_n)_i, p_2(v_1,\ldots, v_n)_i\}
	\end{align*} 
	Then, the mechanisms $M_{max}=(f,p_{max})$ and $M_{min}=(f,p_{min})$ are also moral.
\end{proposition}
\begin{proof}
	We will use the equivalence of moral mechanisms and profit maximizers. Thus, we need to show that the profit maximizers with the payment functions $p_{max}$ and $p_{min}$ also implement $f$. Consider an instance $(v_1,\ldots, v_n)$. We distinguish between the case where the item is allocated ($f(v_1,\ldots, v_n)=i$, for some $i\in N$) and the case where it is not ($f(v_1,\ldots, v_n)=0$). Let $p^1_{-i}(v_{-i}), p_{-i}^2(v_{-i})$ be the potential payments of player $i$ (possibly $\infty)$ given $v_{-i}$ in $M_1,M_2$, respectively. If $f(v_1,\ldots, v_n)=0$ then the potential profit of every player $i$ is non-positive:
	\begin{align*}
	v_i - p^1_{-i}(v_{-i}) \leq 0, ~&~ v_i - p^2_{-i}(v_{-i}) \leq 0
	\end{align*}
	This clearly implies that:
	\begin{align*}
	v_i - p_{-i}^{max}(v_{-i}) \leq 0, ~&~ v_i - p_{-i}^{min}(v_{-i}) \leq 0
	\end{align*}
	Therefore, the potential profit of each player is non-positive in the mechanisms $M_{max},M_{min}$ and the item is not allocated.
	
	Now, consider the case where $f(v_1,\ldots,v_n)=i$, for some $i\in N$. Then, for every $j\neq i$:
	\begin{align*}
	v_i - p^1_{-i}(v_{-i}) \geq 0;~&~~v_i - p^1_{-i}(v_{-i}) \geq v_{j} - p_{-j}^1(v_{-j}) \\
	v_i - p^2_{-i}(v_{-i}) \geq 0;~&~~v_i - p^2_{-i}(v_{-i}) \geq v_{j} - p^2_{-j}(v_{-j})
	\end{align*}
	Note the potential profit of player $i$ is non-negative under both $p_{min}$ and $p_{max}$. It remains to show that this potential profit is at least as large as the potential profit of the other players.
	
	Without loss of generality, suppose that $p^1_{-i}(v_{-i}) \geq p^2_{-i}(v_{-i})$. Notice that for every $j$ for which $p_{-j}^1(v_{-j})\geq p_{-j}^2(v_{-j})$ 
	the potential payments of $M_{max}$ are the same as the potential payments of $M_1$ and the potential payments of $M_{min}$ are the same as the potential payments of $M_2$. 
	We are left with players $j$ for which $p_{-j}^1(v_{-j})<p_{-j}^2(v_{-j})$. In this case the potential payment of $i$ is identical in $M_1$ and $M_{max}$ and the potential payment of $j$ is the same in $M_2$ and $M_{max}$. Now we have that:
	\begin{align*}
	v_i - p^1_{-i}(v_{-i}) \geq v_j - p_{-j}^1(v_{-j}) > v_{j} - p_{-j}^2(v_{-j})
	\end{align*}
	%
	as required. A similar argument shows that $M_{min}$ is also moral.
\end{proof}

\begin{claim} \label{clm-no-char}
	Let $f$ be an allocation function for two players that has the following properties:
	\begin{itemize}
		\item When $v_1+ v_2\geq 1$ and $v_1,v_2\leq 1$, the item is allocated (i.e., $f(v_1,v_2) \in \{1,2\}$).
		\item When $v_1,v_2 >1$, the item is not allocated.
		\item When $v_1>1$, $v_2\leq 1$ the item is allocated to player $1$. Similarly, when $v_2>1$, $v_1\leq 1$ the item is allocated to player $2$.
	\end{itemize}
	There exists a payment function  that implements $f$ as a $1$-moral mechanism.
\end{claim}
\begin{proof}
	We use the following payment function: if $v_i\leq 1$ the price that player $i$ presents to the other player (i.e., $p_{-i}$) is $1-v_i$. If $v_i>1$ the price that player $i$ presents is $\infty$. 
	
	We now show that the profit maximizer that uses this payment function implements $f$. When $v_1+ v_2\geq 1$ and $v_1,v_2\leq 1$, the profit of player $1$ is $v_1-(1-v_2)=v_1+v_2-1$ and the profit of player $2$ is $v_2-(1-v_1)=v_1+v_2-1$. Note that the profits of the players are identical and non-negative. Thus, the profit maximizer can allocate the item to either player $1$ or to player $2$ with no constraints.
	
	If $v_i>1$ the item is obviously not allocated to the other player, since the price is $\infty$. However, if the value of the other player is $v_{j}\leq 1$, then the price presented to player $i$ is $1-v_{j}\leq 1$, hence the profit of player $i$ is positive and player $i$ will be allocated the item. 
\end{proof}

When $\alpha < 1$ we are still far from having a characterization of the allocation functions of $\alpha$-moral mechanisms, but we do have a simple tool to rule out the implementability of some functions:

\begin{claim} \label{clm-rule}
	Let $M=(f,p)$ be an $\alpha$-moral mechanism for some $\alpha<1$. Let $v_{-i,-j}$ be a vector that specifies the values of all players except $i$ and $j$. Then, for every $v_i,v'_i,v_j,v'_j$ it cannot be that all the following equalities simultaneously hold:
	\begin{align*}
	f(v_i,v_j,v_{-i,-j})=i,~f(v'_i,v'_j,v_{-i,-j})=i,~f(v'_i,v_j,v_{-i,-j})=j,~f(v_i,v'_j,v_{-i,-j})=j
	\end{align*}
\end{claim}
\begin{proof}
	Since $M$ is $\alpha$-moral with $\alpha<1$ the following set of inequalities should hold:
	\begin{itemize}
		\item In the instance $(v_i,v_j,v_{-i,-j})$ player $j$ prefers truth-telling: $v_j-p_{-j}(v_i,v_{-i,-j}) <  v_i-p_{-i}(v_j,v_{-i,-j})$.
		\item In the instance $(v_i',v'_j,v_{-i,-j})$ player $j$ prefers truth-telling: $v'_j-p_{-j}(v'_i,v_{-i,-j}) <  v'_i-p_{-i}(v'_j,v_{-i,-j})$.
		\item In the instance $(v_i',v_j,v_{-i,-j})$ player $i$ prefers truth-telling: $v'_i-p_{-i}(v_j,v_{-i,-j}) < v_j-p_{-j}(v'_i,v_{-i,-j}) $.
		\item In the instance $(v_i,v'_j,v_{-i,-j})$ player $i$ prefers truth-telling: $v_i-p_{-i}(v'_j,v_{-i,-j}) < v'_j-p_{-j}(v_i,v_{-i,-j}) $.
	\end{itemize}
	By summing up all inequalities we get that $0<0$, a contradiction.
\end{proof}

\section{Proof of Theorem \ref{thm-standard}}\label{sec-thm-standard}

This section is devoted to proving Theorem \ref{thm-standard}. We gradually characterize the payment functions of a revenue-maximizing moral mechanism. We will show that there exists a revenue-maximizing moral mechanism in which in every instance there is at most one player with positive profit, hence this mechanism is truthful.

Let $V$ denote all the values that are in the support of $\mathcal F$. For a vector $v_{-i}$ that specifies the values of all players except player $i$, we let $p_{-i}(v_{-i})$ denote the potential payment presented to player $i$ when the values of the other players are according to $v_{-i}$. We let $\max(v_{-i})$ denote the maximum value in $v_{-i}$.

We define a partial order $\succ$ over all the  vectors in $V^{n-1}$. For $v \in  V^{n-1}$ we let $sort(v)$ denote the $n-1$ entries of the vector $v$ sorted from largest to smallest. For $v, v' \in  V^{n-1}$ we say that $v \succ v'$ if $sort(v)$ is lexicographically bigger than $sort(v')$. For example, $v \succ v'$ if the highest value in $v$ is greater than the highest value in $v'$.

Next, we state the key proposition in the proof of the theorem. The proposition considers the first (according to $\succ$) $v_{-i}$ such that the potential payment of player $i$ is less than $\max(v_{-i})$. It says that there is another mechanism where $p_{-i}(v_{-i})\geq \max(v_{-i})$ and the payment functions of all vectors before $v_{-i}$ in $\succ$ are the same. Moreover, the revenue of the new mechanism is not lower than that of the original mechanism. Note that the proof uses some useful lemmas that can be found in Subsection \ref{subsubsec-auxillary}.

\begin{proposition} \label{prop-main}
Let $x$ be the monopolist price of $\mathcal F$. Fix some mechanism $M$ with payment functions $p_{-1},\ldots, p_{-n}$. Suppose that for some $v_{-i}\in V^{n-1}$ with $\max(v_{-i})\geq x$ we have that for every $\hv_{-j}\succ v_{-i}$ it holds that $p_{-j}(\hv_{-j})\geq \max(\hv_{-j})$. Then, there exists a mechanism $M'$ with payment functions $p'_{-1},\ldots, p'_{-n}$ such that for every $\hv_{-j}\succ v_{-i}$ it holds that $p'_{-j}(\hv_{-j})\geq \max(\hv_{-j})$. In addition, $p'_{-i}(v_{-i}) \geq \max(v_{-i})$. Moreover, the revenue of $M'$ is not lower than the revenue of $M$.
\end{proposition}
Before proving the proposition we show that the proposition immediately implies a characterization of the payment functions for all vectors $v_{-i}$ with $\max_i(v_{-i})\geq x$:
\begin{corollary}
There exists a revenue maximizing mechanism $M$ with payment functions $p_{-1},\ldots, p_{-n}$ such that for every $v_{-i}$ with $\max(v_{-i})\geq x$ we have that $p_{-i}(v_{-i})\geq \max(v_{-i})$.
\end{corollary}
\begin{proof}
If the corollary does not hold, then for some $v_{-i}$ that satisfies the conditions of the proposition we have that $p_{-i}(v_{-i})< \max(v_{-i})$. If there are several such $v_{-i}$'s, we assume that $v_{-i}$ is a maximal such vector according to $\succ$. Note that by the maximality of $v_{-i}$, we have that for every $\hv_{-j}\succ v_{-i}$, $p_{-j}(\hv_{-j})\geq \max(\hv_{-j})$. The conditions of the proposition thus hold and thus there is another mechanism $M'$ with payment functions $p'_{-1},\ldots, p'_{-n}$ such that $p'_{-i}(v_{-i}) \geq \max(v_{-i})$. We continue similarly to the next vector $v'_{-i}$ that violates the corollary, fixing it by obtaining a new mechanism, and so on.
\end{proof}

We are now ready to prove Proposition \ref{prop-main}:

\begin{proof} 
Let $v_{-i}$ be such that the conditions of the proposition hold but $p_{-i}(v_{-i})<\max(v_{-i})$. Let $v_i=p_{-i}(v_{-i})$. The following claim will be useful:
\begin{claim}\label{claim-only-small-are-profitable}
Let $v_{-i}$ and $v_i$ be as above. In any instance $(z,v_{-i})$, $z\geq v_i$, it holds that if $k$ is a player with non-negative profit then $v_k \geq v_i$.
\end{claim}
\begin{proof}
Let $k$ be a player with $v_k < v_i$. Let $v^z_{-k}$ be the vector that is obtained from $(z,v_{-i})$ by removing the $k$'th coordinate. Observe that if $v_k<v_i$, then $v^z_{-k} \succ v_{-k} \succ v_{-i}$, where the first transition follows since we essentially ``replace'' the value $v_i$ with the higher value $z$ and the second transition follows since $v_k<v_i$. Thus, by the conditions of the proposition $p_{-k}(v^z_{-k})\geq\max(v^z_{-k})\geq v_i$, and player $k$'s profit is negative since $v_k<v_i$.
\end{proof}

We prove the proposition by dividing into cases. Similarly to the proof sketch, we first consider the case where some player $j\neq i$ in the instance $(v_i,v_{-i})$ has a non-negative potential profit, and then consider the case where all players $j\neq i$ have negative potential profit. These two cases are handled by the next two claims.

\begin{claim} \label{clm:smaller-price}
Fix some mechanism $M$ with payment functions $p_{-1},\ldots,p_{-n}$. Let $v_{-i}$ be such that the conditions of Proposition \ref{prop-main} hold in $M$. Let $v_i=p_{-i}(v_{-i})$ and suppose that in the instance $(v_i,v_{-i})$ there exists some player $j\neq i$ such that $v_j\geq p_{-j}(v_{-j})$. Then, there exists a mechanism $M'$ with payment functions $p'_{-1}, \ldots, p'_{-n}$ which are identical to the payment functions $p_{-1},\ldots, ,p_{-n}$ of $M$ except that $p'_{-i}(v_{-i})= p_{-i}(v_{-i})+\eps$ that generates at least as much revenue as $M$. 
\end{claim}

\begin{claim} \label{clm:higher-price}
Fix some mechanism $M$ with payment functions $p_{-1},\ldots,p_{-n}$. Let $v_{-i}$ be such that the conditions of Proposition \ref{prop-main} hold in $M$. Let $v_i=p_{-i}(v_{-i})$ and suppose that in the instance $(v_i,v_{-i})$, for every $k\neq i$, $p_{-k}(v_{-k})>v_k$. Let $M'$ be a mechanism with payment functions $p'_{-1}, \ldots, p'_{-n}$ which are identical to the payment functions $p_{-1},\ldots, ,p_{-n}$ of $M$ except that $p'_{-i}(v_{-i})= p_{-i}(v_{-i})+\eps$ and $p'_{-j}(v_{-j})=v_j$, where $j$ is the player with the highest value in $v_{-i}$. Then, the revenue of $M'$ is not lower than the revenue of $M$. 
\end{claim}
Proofs of these claims can be found in Subsections \ref{subsec-proof-smaller} and \ref{subsec-proof-higher}. As in the outline, the two claims imply the proposition: for $v_{-i}$ defined in the proposition, use one of the claims to obtain a new mechanism where $p_{-i}(v_{-i})$ is increased by $\eps$, and repeat until $p_{-i}(v_{-i})\geq \max(v_{-i})$.
\end{proof}

We have characterized all $p_{-i}(v_{-i})$ when $\max(v_{-i })\geq x$. It remains to characterize the payment functions for vectors $v_{-i}$ such that $\max(v_{-i})<x$. 

\begin{proposition}\label{prop-less-than-mono}
 	Consider a revenue maximizing mechanism for a distribution $\mathcal F$. If for every player $i$ and $v_{-i}$ such that $\max(v_{-i})\geq x$ it holds that $p_{-i}(v_{-i})\geq \max(v_{-i})$, then for every player $i$ and $v_{-i}$ such that $\max(v_{-i})<x$ it holds that $p_{-i}(v_{-i})=x$.
\end{proposition}
\begin{proof}
If our goal is only to maximize the revenue extracted from player $i$ (ignoring the contribution of the other players) the threshold price of player $i$ should be $x$, as $x\in\arg\max p(1-F(p-\eps))$. 

When $p_{-i}(v_{-i})= x$, the revenue of the mechanism (conditioned on the values of all players except player $i$ are according to $v_{-i})$ is indeed $\max_p(p(1-F(p-\eps)))$, all of it due to payments made by player $i$. To see this, let $v^z_{-j}$ be the vector obtained from $(z, v_{-i})$ by removing the $j$'th coordinate, for $z\geq x$. Since $\max(v^z_{-j})=z\geq x>v_{j}$, for every $j\neq i$ (by the conditions of the proposition), we have that only player $i$ is profitable in those instances. Thus, when $p_{-i}(v_{-i})= x$ we extract as much revenue as possible from player $i$, without affecting the payments of the other players. This proves that setting $p_{-i}(v_{-i})=x$ generates more revenue than setting $p_{-i}(v_{-i})=\bar x$, for $\bar x \geq x$.  Note that if we set $p_{-i}(v_{-i})=\underline x$, for $\underline x< x$, we extract less revenue from player $i$ and may charge less from the other players as there might be an additional player who has strictly positive profit which can potentially reduce the payment. We conclude that setting $p_{-i}(v_{-i})= x$ maximizes the revenue.
 \end{proof}

This establishes that for every $v_{_i}$, $p_{-i}(v_{-i})\geq \max(v_{-i})$. This already proves that $M$ is truthful, as needed, since in every instance there is at most one player with positive potential profit. For completeness, we note that Claim \ref{claim-truthful-characterization} provides a specific description of the payment functions.

\subsection{Proof of Claim \ref{clm:smaller-price}}\label{subsec-proof-smaller}

	The setting of this claim is illustrated in Figure \ref{fig-smaller-price} for two players. 
We analyze the difference between the revenue of $M$ and the revenue of $M'$. Since the values of the payment functions of $M$ and $M'$ differ only in $p_{-i}(v_{-i})$ (i.e., $p'_{-i}(v_{-i}) = p_{-i}(v_{-i})+\eps$), it suffices to fix the values of all players except player $i$ according to $v_{-i}$ and analyze the difference in the expected revenue, where expectation is taken over the value of player $i$. We now partition the set of instances of the form $(z,v_{-i})$, as a function of the value of $z$, and analyze the expected revenue of $M$ and $M'$.
\begin{itemize}

	\item \textbf{$z< p_{-i}(v_{-i})$.} Player $i$'s potential profit is negative in $M$ and in $M'$. For each $k\neq i$, $p_{-k}(v^z_{-k})$ is the same in $M$ and in $M'$ and thus there is no difference in the revenue of $M$ and $M'$ in those instances.

	\item \textbf{$z=v_i$.} The potential profit of player $i$ in $M$ is $0$, since $p_{-i}(v_{-i})= v_i$. The potential profit of player $i$ in $M'$ is negative. However, since $p_{-j}(v_{-j})\leq v_{j}$ the profit of player $j$ is nonnegative so the item will be allocated in both $M$ and $M'$. Let $k\neq i$ be a player with the highest value among all player with the highest non-negative potential profit in $M$. By Claim \ref{claim-only-small-are-profitable}, $v_k\geq v_i$. If player $k$'s potential profit is positive the allocation and payment are identical in both $M$ and $M'$. If player $k$'s potential profit is $0$, the allocation might change from player $i$ to player $k$ and collect $v_k$. Since $M$ maximizes revenue, this can happen only if $v_k=v_i$, otherwise both $k$ and $i$ have equal profits and $M$'s tie breaking rule should sell the item to player $k$ and charge $v_k>v_i$. We have that the revenue of $M'$ is at least $v_k$, which is the same as the revenue of $M$.

		\item \textbf{$\max(v_{-i}) \geq  z > v_i$.} Since $z\geq v_i+\eps=p_{-i}(v_{-i})+\eps$, player $i$'s potential profit is non-negative in both mechanisms. Moreover, the potential profit of each player $k$ with $v_k<v_i$ is negative, by Claim \ref{claim-only-small-are-profitable}. By Observation \ref{obs-small-is profitable}, the revenue of $M'$ is at least the revenue of $M$ in each instance $(z,v_{-i})$.
	
		\item \textbf{$z > \max(v_{-i})$.} We claim that the revenue in the instance $(z,v_{-i})$ is larger by $\eps$ in $M'$. For $k\neq i$, let $v^z_{-k}$ be $(z,v_{-i})$ without the $k$'th coordinate. Observe that $\max(v^z_{-k})>\max(v_{-i})$. Therefore, by the conditions of the proposition, for every $k\neq i$, $p_{-k}(v^z_{-k}) > \max(v_{-i})$, since $v^z_{-k}\succ v_{-i}$. Hence, all players but $i$ are not profitable. Now, observe that $z >\max(v_{-i})\geq v_i+\eps=p_{-i}(v_{-i})+\eps$. The payment of player $i$ is $v_i+\eps$ in $M'$ and $v_i$ in $M$, hence player $i$ is the only profitable player in this instance in both mechanisms. In $M'$ player $i$'s payment increases by $\eps$, thus the revenue of $M'$ is larger.
		
\end{itemize}
We conclude that in all possible scenarios the revenue of $M'$ is not smaller.


\subsection{Proof of Claim \ref{clm:higher-price}}\label{subsec-proof-higher}

	The setting of this claim is illustrated in Figure \ref{fig-higher-price} for two players. 
%
We first claim that $p_{-j}(v_{-j})= v_j+\eps$ (recall that $j$ is the player with the highest value in $v_{-i}$). To see this, observe that by the assumptions of the claim, for every $z\geq  v_j$ all players $k\neq j$ have negative potential profit in the instance $(z,v_{-j})$. This is because 
	for every player $k\neq j $ and $z$ as above it holds that $v^z_{-k} \succ v_{-k}$, where $v^z_{-k}$ denotes the vector obtained from $(z, v_{-j})$ by removing the $k$'th coordinate. Hence by the conditions of Proposition \ref{prop-main},  $p_{-k}(v^z_{-k})\geq\max(v^z_{-k})$. Moreover, by construction $\max(v^z_{-k})=z> \max(v_{-i})\geq v_k$ and hence we get that player $k$'s potential profit is negative in $(z,v_{-j})$. Next, we want to apply Lemma \ref{lem-rev-max}. Intuitively, the Lemma builds on the fact that in a single-bidder auction whose value is drawn from some regular distribution, the revenue increases as the take-it-or-leave-it price is closer to the monopolist price. The lemma generalizes this property for multi-player auctions under certain conditions, which hold in our case. Applying Lemma \ref{lem-rev-max} together with the condition that $p_{-j}(v_{-j})>v_{-j}$ we get that $p_{-j}(v_{-j})=\max(v_{-i})+\eps = v_j + \eps$.

The heart of the proof is the following lemma:
\begin{lemma} \label{lem-sim-change}
Let $M$ be a mechanism for players that their values are independently sampled from a standard distribution $\mathcal F$ with payment functions $p_{-1},\ldots, p_{-n}$. Consider an instance $(v_1,\ldots, v_n)$. Let $j$ be such that $v_j=\max(v_{-i})$. Suppose that:
\begin{itemize}
\item $p_{-i}(v_{-i})=v_i<\max(v_{-i})$.
\item $p_{-j}(v_{-j})=v_j+\eps$.
\item Only player $j$ might have a non-negative profit in the instance $(z,v_{-j})$, $z>v_j$.
\item Only players $k$ with $v_k\geq v_i$ might have non-negative profit in the instance $(z,v_{-i})$, $z>v_i$.
\end{itemize}
Then, the mechanism $M'$ which has payment functions $p'_{-1},\ldots,  p'_{-n}$ that are identical to $p_{-1}, \ldots, p_{-n}$ except that $p'_{-i}(v_{-i})=v_i+\eps$ and $p'_{-j}(v_{-j})=v_j$ generates at least as much revenue as $M$.
\end{lemma}

Before proving the lemma, note that its conditions hold: we have already established that the first three conditions hold. The last condition holds by Claim \ref{claim-only-small-are-profitable}. Thus, by Lemma \ref{lem-sim-change} we have that the revenue of $M'$ is at least as large as that of $M$.

\begin{proof}(of Lemma \ref{lem-sim-change})
Let $F,f$ be the CDF and PDF of $\mathcal F$, respectively. The only instances that the revenue of $M$ and $M'$ might differ in are those in which the values of all players except player $i$ are according to $v_{-i}$ or the values of all players except player $j$ are according to $v_{-j}$. The expected revenue of $M$ in those instances is:
	\begin{align} \label{eq-swap-p}
	\prod_{k\neq i}f(v_{k})&\int_{0}^{1} \delta(z\neq v_i) \cdot f(z) \cdot  Rev^M(z,v_{-i})dz \\
	&+ \prod_{k\neq j}f(v_{k})\int_{0}^{1} \delta(z\neq v_j) \cdot f(z) \cdot Rev^M(z,v_{-j})dz + \prod_{k}f(v_{k})\cdot v_i\nonumber
	\end{align}
	where by $Rev^M(z,v_{-i}) $ we mean the revenue of $M$ in the instance $(z,v{-i})$ and $\delta(\cdot)$ is an indicator function that gets a boolean condition and returns $1$ if the condition is true and $0$ otherwise. Similarly, the expected revenue of $M'$ is:
	\begin{align}  \label{eq-swap-p'}
		\prod_{k\neq i}f(v_{k})&\int_{0}^{1} \delta(z\neq v_i) \cdot f(z) \cdot  Rev^{M'}(z,v_{-i})dz \\
		&+ \prod_{k\neq j}f(v_{k})\int_{0}^{1} \delta(z\neq v_j) \cdot f(z) \cdot Rev^{M'}(z,v_{-j})dz + \prod_{k} f(v_{k})\cdot v_j\nonumber
	\end{align}
	We would like to show that the expected revenue of $M'$ is at least as high as that of $M$  ($(\ref{eq-swap-p}) \leq (\ref{eq-swap-p'})$). By rearranging we get that it is sufficient to show that:
	\begin{align} \label{eq3}
	f(v_{j})&\int_{0}^{1} \delta(z\neq v_i) \cdot f(z) \cdot \Big(  Rev^M(z,v_{-i})-Rev^{M'}(z,v_{-i}) \Big)dz  \leq \\  
	&f(v_{i})\int_{0}^{1} \delta(z\neq v_j) \cdot f(z) \cdot \Big( Rev^{M'}(z,v_{-j}) - Rev^M(z,v_{-j})\Big)dz 
	+ f(v_j)\cdot f(v_i) \cdot (v_j-v_i) \nonumber
	\end{align}
		Notice that	in every instance ($z, v_{-j})$, $z>v_j$, only player $j$ might be allocated. In these instances, the revenue of $M$ will be bigger by $\eps$ than that of $M'$. Furthermore, in every instance $(z,v_{-i})$, $z < v_j$, player $j$ will not be allocated since $p_{-j}(v_{-i})>v>z$. In any such instance potential profits of the remaining players are the same in $M$ and $M'$ and thus the revenue in these instances is identical in both $M$ and $M'$ (recall that player $j$ has a negative profit in both mechanisms). Thus, it holds that: 
	\begin{align*}
	f(v_{i})\int_{0}^{1} \delta(z\neq v_i)  \cdot f(z) \cdot  \Big(  Rev^M(z,v_{-i})-Rev^{M'}(z,v_{-i}) \Big)dz 
	=	f(v_i) \cdot \eps \cdot (1-F(v_j))
	\end{align*}
	Consider instances of the form $(z,v_{-j})$. When $z> v_j$, only player $j$ is profitable. Since $p'_{-j}(v_{-j})-p_{-j}(v_{-j})=\eps$, $M'$ gains an extra revenue of $\eps$ comparing to $M$ in these instances. When $v_j \geq z > v_i$, since the value of all players with non-negative profit is at least $v_i$, Observation \ref{obs-small-is profitable} implies that the revenue can only increase in $M'$.  Finally, when $z<v_i$, player $i$ is not allocated and the potential profits are the same in both $M$ and $M'$, thus the revenue is identical in $M$ and $M'$. Therefore:
	\begin{align*}
	f(v_j)\int_{0}^{1} \delta(z\neq v_j) \Big( Rev^{M'}(z,v_{-j}) - Rev^M(z,v_{-j})\Big)dz   \geq  f(v_j) \cdot \eps  \cdot(1-F(v_j))
	\end{align*}
	Putting this together, we get that  a sufficient condition for inequality (\ref{eq3}) to hold  is:
	\begin{align} \label{eq4}
	f(v_i) \cdot \eps \cdot (1-F(v_j)) \leq f(v_j) \cdot \eps  \cdot(1-F(v_j)) + f(v_j)\cdot f(v_i) \cdot (v_j-v_i)
	\end{align}
Finally, note that inequality (\ref{eq4}) holds since  $\mathcal F$ is standard.
\end{proof}



This concludes the proof of Claim \ref{clm:higher-price}.

\subsection{Auxiliary Lemmas}\label{subsubsec-auxillary}

\begin{observation} \label{obs-small-is profitable}
Consider an instance $(v_1,\ldots, v_n)$ and some player $i$ such that $t=p_{-i}(v_{-i})<v_i$. Suppose that for every player $j$ with $v_j<v_i$ we have that $p_{-j}(v_{-j})> v_j$. Then the revenue in the instance $(v_1,\ldots, v_n)$ is at least as high when $p_{-i}(v_{-i})=t+\eps$ as when $p_{-i}(v_{-i})=t$.
\end{observation}
\begin{proof}
In a profit maximizer, if some player $j$, $j\neq i$ wins when $p_{-i}(v_{-i})=t$ then he also wins when $p_{-i}(v_{-i})=t+\eps$. This is because player $j$'s profit is the same but player $i$'s profit has decreased: $v_j-p_{-j}(v_{-j}) \geq v_i -t \implies v_j-p_{-j}(v_{-j}) \geq v_i -t-\eps$. The revenue is identical in both cases. 
	
Suppose that player $i$ wins when $p_{-i}(v_{-i})=t$ (the revenue is obviously $t$). Hence, for every player $j$ with non-negative profit, $v_j-p_{-j}(v_{-j}) \leq v_i -t \implies p_{-j}(v_{-j}) \geq t+v_j-v_i \geq t$, where the last transition is because $v_j\geq v_{i}$. Player $i$'s potential profit is non-negative in both cases, thus the item will be allocated in both and the payment will at least $t$ (either $p_{-j}(v_{-j}) \geq t$ or $t+\eps$).
\end{proof}

\begin{lemma} \label{lem-rev-max}
Consider a revenue maximizing mechanism $M$ with payment functions $p_{-1},\ldots, ,p_{-n}$ where values are sampled independently from a standard distribution $\mathcal F$. Suppose that for some $v_{-i}$ and $t\geq x$ it holds that all players except player $i$ have negative profit in every instance $(z,v_{-i})$, $z\geq t$, then $p_{-i}(v_{-i})\leq t$.
%
%
%
\end{lemma}
\begin{proof}
Let $F, f$ be the CDF and PDF of $\mathcal F$, respectively. Towards contradiction, let $v_{-i}$, $t$ be such that all the conditions hold but $p_{-i}(v_{-i})>t$. We will show that the mechanism $M'$ with identical payment functions $p'_{-1}, \ldots ,p'_{-n}$ except that $p'_{-i}(v_{-i})=p_{-i}(v_{-i})-\eps$ generates more revenue.
	 
The only difference in the revenue of $M$ and $M'$ is in instances $(z,v_{-i})$, $z\geq t$. In these instances, by the conditions of the lemma, only player $i$ might be profitable in the instance $(z,v_{-i})$ in both mechanisms. Thus, we only need to compare the difference in the revenue in these instances. Let $v_i=p_{-i}(v_{-i})$ and let $Rev^M(z,v_{-i})$ be the revenue of $M$ in the instance $(z,v_{-i})$. In these instances:

	\begin{itemize}	
		\item  The revenue of $M$ is: $\int_t^1 f(z)Rev^M(z,v_{-i})dz =v_i\cdot (1-F(v_i-\eps))$.

		\item The revenue of $M'$ is: $\int_t^1 f(z)Rev^{M'}(z,v_{-i})dz =(v_i-\eps)\cdot (1-F(v_i-\eps)+f(v_i-\eps))$.

%
	\end{itemize}

Therefore, the revenue of $M'$ is at least that of $M$ if: 
\begin{align*}
(v_i-\eps)\cdot (1-F(v_i-\eps)+f(v_i-\eps)) &\geq v_i\cdot (1-F(v_i-\eps))\\
-\eps (1-F(v_i-\eps))+(v_i-\eps)f(v_i-\eps) &\geq 0 \\
v_i-\eps - \frac {\eps \cdot (1-F(v_i-\eps))} {f(v_i-\eps)} &\geq 0\\
\end{align*}
I.e., $\vv(v_i-\eps)\geq 0$, which is true since $\mathcal F$ is a standard distribution and $v_i> x$.
%
%
%
	\end{proof}
}

\end{document}